\documentclass[12pt,reqno,a4paper]{amsart}
\usepackage{amssymb,latexsym}
\usepackage{amsthm, amsmath, amsfonts}
\usepackage{hyperref}
\usepackage{empheq}
\usepackage{chngcntr}
\usepackage{apptools}
\usepackage{graphicx}  
\usepackage{dcolumn}   
\usepackage{bm}        
\usepackage{enumerate}
\usepackage{slashed}
\usepackage{simplewick}
\usepackage{comment}
\usepackage{color}
\usepackage{soul}
\usepackage[english]{babel}
\usepackage{appendix}
\usepackage[utf8]{inputenc}
\usepackage{tikz} 
\usetikzlibrary{decorations.markings}
\usepackage{enumitem}
\numberwithin{equation}{section}

\usepackage{marginnote}
\usepackage[normalem]{ulem}

\def\be{\begin{equation}}
\def\ee{\end{equation}}
\def\bea{\begin{eqnarray}}
\def\eea{\end{eqnarray}}
\def\bz{\bar z}

\def\p{\partial}
\def\bp{\bar\partial}

\def\b{\beta}

\def\bcal_k{\mathcal B_k}

\newcommand{\NPhi}{{N_{\phi}}}

\def\({\left(}
\def\){\right)}

\def\g{{\rm g}}

\def\Im{{\rm Im\,}}
\def\m{{\rm m}}

\def\L{{\rm L}}
\def\eq{Eq.\ \eqref}

\newtheorem{maindefn}{{\sc Definition}}
\newtheorem*{maindefn*}{{\sc Definition}}
\newtheorem{theo}{{\sc Theorem}}[section]

\newtheorem{prop}[theo]{{\sc Proposition}}

\newenvironment{rem}{\medskip\noindent{\it Remark:\/} }{\medskip}

\newtheoremstyle{dotless}{}{}{\itshape}{}{\bfseries}{}{ }{}
\theoremstyle{dotless}

\begin{document}

\title[Laughlin states on higher genus Riemann surfaces]
{Laughlin states on higher genus Riemann surfaces}
\author[Semyon Klevtsov]{Semyon Klevtsov}

\maketitle

\begin{center}
{\small
\address{\it Universit{\"a}t zu K{\"o}ln,
Mathematisches Institut,\\ Weyertal 86-90, 50931 K{\"o}ln, Germany}
}
\end{center}

\begin{abstract} 
Considering quantum Hall states on geometric backgrounds has proved over the past few years to be a useful tool for uncovering their less evident properties, such as gravitational and electromagnetic responses, topological phases and novel geometric adiabatic transport coefficients. One of the transport coefficients, the central charge associated with the gravitational anomaly, appears as a Chern number for the adiabatic transport on the moduli spaces of higher genus Riemann surfaces. This calls for a better understanding of the QH states on these backgrounds. Here we present a rigorous definition and give a detailed account of the construction of Laughlin states on Riemann surfaces of genus $\g>1$. 
By the first principles construction we prove that the dimension of the vector space of Laughlin states is at least $\beta^\g$ for the filling fraction $\nu=1/\beta$. Then using the  path integral for the 2d bosonic field compactified on a circle, we reproduce the conjectured $\beta^\g$-degeneracy as the number of independent holomorphic blocks. We also discuss the lowest Landau level, integer QH state and its relation to the bosonization formulas on higher genus Riemann surfaces.
\end{abstract}

\thispagestyle{empty}
\tableofcontents


\section{Introduction}

Laughlin state \cite{L} is a trial $N$ particle wave function successfully approximating the ground state of the fractional quantum Hall effect for the values of the Hall conductance given by simple fractions $\nu=1/\beta$. Considering the Laughlin states and QH states for other filling fractions in geometric settings, e.g., on compact Riemann surfaces is a helpful tool for understanding their geometric and topological properties. While on the sphere the Laughlin ground state is unique \cite{H}, on the torus its $\beta$-fold degeneracy \cite{HR,YHL} provided one of the first examples of topological phases of matter. The $\beta^\g$-fold degeneracy of the Laughlin states on the genus $\g>1$ surface was first conjectured in Ref.\ \cite{WN} using the arguments from Chern-Simons theory, see also \cite{W1989} for related work on the topological degeneracy of chiral spin states on Riemann surfaces. 

Besides the topological degeneracy, another important application of QH states on Riemann surfaces is the geometric adiabatic transport on the associated moduli spaces, such as the moduli space of flat connections, complex structures, etc. This goes back to the observation that the Hall conductance is the first Chern class of the vector bundle of QH states over the moduli space of flat connections, i.e., Aharonov-Bohm phases \cite{TW, T2,ASZ}. The adiabatic transport on the moduli space of elliptic curves led to the discovery of another Chern number -- the odd Hall viscosity \cite{ASZ1, L1, R, TV2}. More recently the interest in the subject is connected with the investigation of the gravitational anomaly that QH states exhibit when coupled to a curved spatial metric, see e.g. Refs.\ \cite{K,CLW,FK,CLW1,LCW,BR1,KW,KMMW} for 2d and Refs.\ \cite{Son,AG1,AG2,BR,AG4,KMMW} for 2+1d perspectives. The adiabatic transport on the moduli space $\mathcal M_{\g}$ of complex structures of Riemann surface of $\g>1$ involves a novel transport coefficient, called the central charge \cite{KW,BR1,KMMW} associated with the gravitational anomaly. The central charge also controls the behaviour of QH states on singular surfaces \cite{CCLW,Kl16}, adiabatic transport on the moduli space $\mathcal M_{0,n}$ of sphere with $n$ punctures \cite{C16,CW17} and braiding of geometric defects in QH states \cite{BJQ,Gr}.  

In the present paper we construct the Laughlin states on a compact Riemann surface $\Sigma$ of genus $\g>1$. As we have already mentioned the early works on higher genus Laughlin states include Ref.\ \cite{WN}, where $\beta^\g$-fold degeneracy was conjectured and  Ref.\ \cite{BN} where closely related conformal blocks corresponding to geometric quantisation of $U(1)$ Chern-Simons theory were constructed. The construction of Laughlin states is also sketched in Ref.\ \cite{IL} (where a different degeneracy of the ground states was reported) and more recently in Ref.\ \cite{Gr}. Our goal here is to provide a detailed account of the construction, starting from one-particle states on $\Sigma$, then explaining the $\beta=1$ integer QH state via the bosonization formula and finally constructing the Laughlin states. 

In the original form \cite{L} the Laughlin wave function is defined for $N$ particles on the complex plane by the formula
\be\label{LaS}
\Psi(z_1,...,z_N)= \prod_{n<m}^N(z_n-z_m)^\beta\;\cdot e^{-\frac B4\sum_{n=1}^N|z_n|^2},\quad \beta\in\mathbb Z_+
\ee
where $B$ is the constant magnetic field. This wave function can be taken to the Riemann surface $\Sigma$ of any genus, where the appropriate formulation is in the language of line bundles and divisors. For the background on line bundles on Riemann surfaces we refer to \cite{Bost}. 

In physics terminology, the magnetic flux (often denoted $\phi$) through a surface is the integral of the normal component of the magnetic field $B$ on the surface. In the case of a closed surface this is quantized $\phi=2\pi \NPhi$, where the integer $\NPhi\in\mathbb Z_+$ is called the number of magnetic flux quanta penetrating the surface. Mathematically, the magnetic field with the flux $\NPhi$ is described by a magnetic field divisor \eq{degDm} of degree $\NPhi$, i.e., a formal sum of points on $\Sigma$,
$
D_\m=\sum_{a=1}^\NPhi q_a,
$
where we allow for coincident points. The one-particle wave functions on the lowest Landau level (LLL) are the holomorphic sections of the line bundle $\mathcal O(D_\m)$, see e.g. \cite{DK} for an extensive discussion. Let $s(z)$ be a holomorphic section with the zeroes located at the points $q_a$ with order of vanishing according to the multiplicity of the point in $D_\m$. All the degenerate LLL can then be constructed by multiplying $f\cdot s$ by the basis of meromorphic functions $f(z)\in H^0(\Sigma,\mathcal O(D_\m))$ \eqref{HO} with the divisor of poles contained in $D_\m$. 

In addition to the magnetic field divisor, the particles are in principle allowed to have an integer or half-integer gravitational spin $s$, via the coupling to the $2s$th tensor power of the spin line bundle $S_\delta^{2s}$, \eq{degDd} with the associated divisor $D_\delta$. In the case of the half-integer $s$, the corresponding wave functions will also be characterized by a choice of the spin structure $\delta$, see \S \ref{sstr}. Finally, magnetic field line bundle $\mathcal O(D_\m)$ can also include the flat component $L_\varphi$, corresponding to AB phases around the one-cycles of the Riemann surface. This line bundle is parameterized by a coordinate $\varphi$ on the Jacobian variety of the Riemann surface, \eq{varphi}. This data describes the line bundle $\rm L$ \eq{rmL} of total degree $\deg\L=\NPhi+2s(\g-1)$. 

The number of one-particle states on the LLL is then counted by the Riemann-Roch theorem \eqref{RRoch} as the dimension of the space of holomorphic sections,
\be\label{RRLLL}
\dim H^0(\Sigma,\rm L) =\deg\L+1-\g,
\ee
assuming $\deg\L$ is sufficiently large, see \eq{bound}. The Laughlin ground state corresponds to the situation of $N$ interacting (via repulsive Coulomb force) particles. The effect of the interaction manifests itself as follows. While the $N$ particle wave function is still holomorphic, the degree of $\L$ is assumed to be divisible by $\beta\in\mathbb Z_+$ and the number of particles is given by
$$
N=\frac1\beta\deg\L+1-\g.
$$
In physics terms one can loosely say that only about a fraction of available LLL states \eq{RRLLL} is occupied. Then the effect of the interaction is accounted for by a certain vanishing condition on the wave function, when any two of the particles meet.

In this setting we define the (holomorphic part of the) Laughlin state on $\Sigma^N=\Sigma\times\cdots\times\Sigma$ with coordinates $z_1,...,z_N$ as follows.
\begin{maindefn*} {\bf \ref{Lsdef}}
A Laughlin state $F(z_1,...,z_N)$ for the filling fraction $1/\beta$ on $\Sigma^N$ satisfies the following conditions:
\begin{enumerate}
\item Restriction $\pi_n F(z)$ to the $n$th component $\Sigma_n$ in $\Sigma^N$, i.e., $F(...,z_n,...)$ with all $z$'s but $z_n$ fixed, transforms as a holomorphic section of the line bundle $\L$, \eq{rmL}.
\item For any pair $z_n$, $z_m$, $n\neq m$, the restriction $\pi_{nm}F$ to $\Sigma_n\times\Sigma_m$ vanishes exactly to the order $\beta$ near the diagonal
\be\nonumber
\pi_{nm}F\sim (z_n-z_m)^\beta,
\ee
in terms of a local complex coordinate on $\Sigma$.
\item For $\beta$ odd $F$ is completely anti-symmetric and for $\beta$ even, completely symmetric on $\Sigma^N$.
\end{enumerate}
\end{maindefn*}
Since for any fixed $\beta$ the sum of any two Laughlin states is again a Laughlin state they form a vector space, which we denote as $\mathbb V_\beta$.
The first property above encodes the polynomial character of the holomorphic pre-factor in \eq{LaS} and the second property encodes the local vanishing properties, which shall be independent of the global geometry. The third property suggest that even and odd values of $\beta$ can be considered at once. Finally the exponential gaussian in \eq{LaS} corresponds to the hermitian norm on the line bundle $\rm L$, introduced in \eq{norm1}.

In Prop.\ \ref{indepst} we show that the following $\beta^\g$ Laughlin states are linearly independent, provided $N\geqslant\g$,
\begin{align}\nonumber
F_r(z_1,...,z_N)=\vartheta\left[\begin{array}{c}\scriptstyle\frac r\beta\\\scriptstyle0\end{array}\right]\Big(\beta{\textstyle \sum_{n=1}^N} z_n-\beta\Delta- {\rm div}\,\L,\beta\tau\Big)
\cdot \prod_{n<m}^N{E(z_n,z_m)}^\beta\cdot\prod_{n=1}^N\sigma(z_n)^{\frac1\g(\beta N+2s-\beta)}.
\end{align}
Here the multi-index $r\in(\mathbb N^+_{\leqslant\beta})^\g$ takes $\beta^\g$ different values $r=(1,...,\beta)^\g$, $\Delta$ is the vector of Riemann constants defined in \eq{riemannC} and the divisor ${\rm div}\, \L$ is given by \eq{divL}. 
The theta function with characteristics and the Prime form $E(z,y)$ are introduced in \S \ref{thet} and \S \ref{primeform} and $\sigma(z)$ is the holomorphic $\frac\g2$-differential introduced in \eq{sigmafunc}, following Ref.\ \cite{Fay92}.

In addition to the first principles construction from the definition above, our goal here is also to derive the Laughlin states from the path integral representation $\mathcal V\bigl(g,B,\{z_n\}\bigr)$, \eq{pathi}, via the correlation function of a string of vertex operators of the compactified boson in 2d. This is the well-known vertex operator construction of Ref.\ \cite{MR} applied in the geometric setting \cite{FK,KW}. Thus, assuming the Def.\ \ref{Lsdef} holds, we show that $\dim \mathbb V_\beta$ is at least $\beta^\g$, however we do not prove it is exactly equal to $\beta^\g$. 

We derive the following result for the path integral, Prop. \ref{nugB},
\be\nonumber
\mathcal V\bigl(g,B,\{z_n\}\bigr)=\mathcal N\cdot
\sum_{\varepsilon',\varepsilon''\in\{0,\frac12\}^\g}\;\sum_{r\in(\mathbb N^+_{\leqslant\beta})^\g}\;e^{4\pi i(\varepsilon',\varepsilon'')}\Vert F_r^{\varepsilon'+\delta',\varepsilon''+\beta\delta''}(z_1,...,z_N)\Vert_h^2,
\ee
where $\mathcal N$ is a certain $z$-independent factor. The summation here goes over the $4^\g$ different choices of the half-integer theta characteristics $(\varepsilon',\varepsilon'')$ relative to the reference spin structure $(\delta',\beta\delta'')$. There are $\beta^\g$ Laughlin states $F_r^{\varepsilon',\varepsilon''}$ for each choice of the spin structure, given simply by the shifts of $F_r$ by half integer points in $Jac(\Sigma)$, see \eq{laughlinst}. 

In contrast to the situation on the torus, the $\beta^\g$-fold degeneracy here does not correspond to the center-of-mass motion, as the latter is not well defined since $\Sigma$ is not an Abelian variety for $\g>1$. The degeneracy arises instead from the center-of-mass motion in the Jacobian variety, the $2\g$-dimensional torus $Jac(\Sigma)$, where $\Sigma^N$ is holomorphically embedded to by the Abel map \eq{abel}. The condition $N\geqslant\g$ ensures that the Abel map $\Sigma^N\hookrightarrow Jac(\Sigma)$ is onto by the Jacobi inversion theorem \cite[III.6.6]{FKra}. Otherwise the image of $\Sigma^N$ will have a codimension greater than zero and the degeneracy of Laughlin states will be less than $\beta^\g$. However, the quantum Hall regime implies $N\gg1$, so for our purposes here we assume the bound $N\geqslant \g$.

In \S \ref{intQH} we use the Fay's identity, \eq{boso}, to demonstrate that at $\beta=1$ the Laughlin states reduce to the Slater determinant of the one-particle states, \eq{LLLs},
\be\nonumber
\det s_n(z_m)|_{n,m=1}^N=C\cdot\vartheta\Big({\textstyle \sum_{n=1}^N} z_n-\Delta-{\rm div}\,\L,\tau\Big)\cdot \prod_{n<m}^NE(z_n,z_m)\cdot\prod_{n=1}^N\sigma(z_n)^{\frac1\g(N+2s-1)}.
\ee
We then explain how this identity arises from the bosonization formula of Refs.\ \cite{Fal,ABMNV,Kn1986,VV} relating the bosonic path integral $\mathcal V\bigl(g,B,\{z_n\}\bigr)$ to a correlation function of spin-$s$ fermions \eq{spinsferm}. 

As a final remark we point out that the path integral construction of sec.\ \ref{freef} corresponds to a CFT with a background charge and with the addition of the magnetic field term. The role of the latter  is to dress the conformal blocks with the hermitian metric of the magnetic line bundle. Then the partition function is not strictly speaking conformally-invariant and neither are the Laughlin states, since they correspond to a gapped ground state and depend the magnetic length. Interestingly, in one special case the conformal invariance in the bulk can be formally restored. This happens when the magnetic field divisor $D_\m$ is a multiple of the canonical divisor and $B$ is proportional to the scalar curvature $R$. 
\\

{\bf Acknowledgements}. 
I would like to thank Peter Zograf and Alexey Kokotov for useful discussions and the anonymous referees for helpful comments.
This work was partially funded by the Deutsche Forschungsgemeinschaft (DFG) -- project number 376817586, SFB TRR 191, RFBR grant 18-01-00926 and QM2 collaboration at the University of Cologne. The author also gratefully acknowledges the hospitality of Anton Alexeev and the University of Geneva during the initial stages of this project. 

\section{Background}

\subsection{Higher genus Riemann surfaces}

We begin with a collection of relevant facts from the theory of Riemann surfaces, Riemann theta functions and holomorphic line bundles. The standard references here are \cite{FKra,M,Fay73} as well as the reviews \cite{Bost,DP}. We mainly focus on the facts that will be used later in the text, in order to make the presentation relatively self-contained.

We consider a compact oriented genus-$\g$ Riemann surface $\Sigma$ and we will be mainly concerned with the case $\g>1$. Although a specific construction will not be of a particular importance, one standard way to describe $\Sigma$ is by the quotient $\mathbb H/\Gamma$, where $\mathbb H$ is the upper half plane and $\Gamma\in PSL(2,\mathbb R)$.

The homology group $H_1(\Sigma,\mathbb Z)$ is generated by $2\g$ one-cycles $a_1,...,a_{\g},b_1,...,b_{\g}$. The canonical choice of the basis in $H_1(\Sigma,\mathbb Z)$ is the symplectic basis, where the intersection numbers are given by
\begin{equation}\label{sympl}
a_j{\scriptstyle\#}\, b_l=\delta_{jl},\quad a_j{\scriptstyle\#}\,  a_l=b_j{\scriptstyle\#}\,  b_l=0,\quad j,l=1,...,\g,
\end{equation}
i.e., the intersection number of the cycles $a_j$ and $b_j$ is $\pm1$ depending on the orientation, while all other intersection numbers vanish, see Fig.\ \ref{fig:highergenus}.
The choice of a canonical homology basis together with a choice of the base point $P_0\in\Sigma$ defines a marking of the Riemann surface.

\begin{figure}[h]
\begin{center}
\begin{tikzpicture}[smooth cycle,scale=0.35]
\tikzset{->-/.style={decoration={
  markings,
  mark=at position .5 with {\arrow{>}}},postaction={decorate}}}
  \tikzset{-<-/.style={decoration={
  markings,
  mark=at position .5 with {\arrow{<}}},postaction={decorate}}}
\draw [black,very thick,bend right] (6,6.35) arc (50:310:3.5);
\draw [black,very thick,bend right] (10+9,1) arc (-130:130:3.5);
\draw [black,very thick,bend right] (10.01,0.99) edge (5.96,.96);
\draw [black,very thick,bend right] (6,6.35) edge (10.04,6.39);
\draw [black,very thick,bend left] (10+5,0.97) edge (6.0+4,.99);
\draw [black,very thick,bend left] (6+4,6.35) edge (10.04+5,6.39);
\draw [black,very thick,bend right] (11,3.7) edge (13.8,3.7);
\draw [black,very thick,bend left] (11.3,3.55) edge (13.5,3.55);
\draw [black,dashed,very thick,bend right] (10.01+9,0.99) edge (5.96+9,.96);
\draw [black,dashed,very thick,bend right] (6.1+9,6.35) edge (10.04+9,6.39);
\draw [black,very thick,bend right] (2,3.7) edge (4.8,3.7);
\draw [black,very thick,bend left] (2.3,3.55) edge (4.5,3.55);
\draw [black,very thick,bend right] (11+9,3.7) edge (13.8+9,3.7);
\draw [black,very thick,bend left] (11.3+9,3.55) edge (13.5+9,3.55);
\draw [red,very thick,bend left,-<-] (3.5,3.5) ellipse (70pt and 35pt);
\draw [green,very thick,bend left,-<-] (3.5+9,3.5) ellipse (70pt and 35pt);
\draw [violet,very thick,bend left,-<-] (3.5+18,3.5) ellipse (70pt and 35pt);
\draw [blue,very thick,bend right]  (3.5,0.2) edge [-<-] (3.5,3.3);
\draw [blue,dashed,very thick,bend right]  (3.5,3.3) edge [-<-] (3.5,0.2);
\draw [cyan,very thick,bend right]  (3.5+9,0.2) edge [-<-] (3.5+9,3.3);
\draw [cyan,dashed,very thick,bend right]  (3.5+9,3.3) edge [-<-] (3.5+9,0.2);
\draw [orange,very thick,bend right]  (3.5+18,0.2) edge [-<-] (3.5+18,3.3);
\draw [orange,dashed,very thick,bend right]  (3.5+18,3.3) edge [-<-] (3.5+18,0.2);        \draw[black] (4,4.6) node [above]{$b_1$};
\draw[black] (4.8,.5) node [above]{$a_1$};
\draw[black] (4+9,4.6) node [above]{$b_2$};
\draw[black] (4.8+9,.5) node [above]{$a_2$};
\draw[black] (4+18,4.6) node [above]{$b_\g$};
\draw[black] (4.8+18,.5) node [above]{$a_\g$};
        \draw[black] (8.15,3.55) node [above]{$P_0$};
                \fill [black] (8,3.6) circle (3pt);
\end{tikzpicture}
{\small \caption{Marking of the Riemann surface.}
\label{fig:highergenus}}
\end{center}
\end{figure}
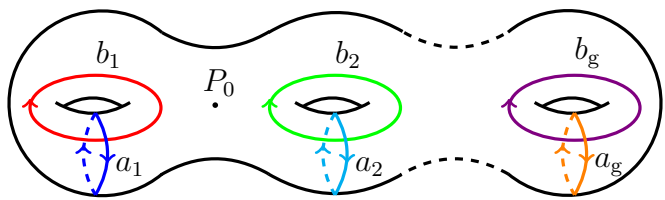

A harmonic one-form $\alpha$ satisfies
$d\alpha=d^c\alpha=0$, where the exterior derivative $d$ and its adjoint $d^c$ are defined as 
\be\label{ddc}
d=dz\wedge\p_z+d\bz\wedge\p_{\bz},\quad
d^c=id\bz\wedge\p_{\bz}-idz\wedge\p_z.
\ee
Harmonic forms provide a basis for the first cohomology group $H^1(\Sigma,\mathbb R)$. Let $\alpha_j,\beta_j\in H^1(\Sigma,\mathbb Z)$ be the basis of harmonic one-forms 
dual to the symplectic basis \eq{sympl} of one-cycles, i.e., normalized as 
\begin{equation}\label{harmf}
\int_{a_j}\alpha_l=\int_{b_j}\beta_l=\delta_{jl},\quad\int_{a_j}\beta_l=\int_{b_j}\alpha_l=0.
\end{equation}

The dimension of the vector space of holomorphic differentials on $\Sigma$ is equal to the genus of the surface $\g$. The canonical basis $\omega_l$ of holomorphic differentials can always be normalized so that 
\begin{equation}
\label{acycl}
\int_{a_j}\omega_l=\delta_{jl},\quad
\int_{b_j}\omega_l=\tau_{jl}
\end{equation}
and the matrix $\tau$ is called the period matrix of $\Sigma$. The period matrix is symmetric and its imaginary part is positive-definite
$$
\tau_{jl}=\tau_{lj},\quad \Im\tau>0. 
$$
This is an immediate consequence of the Riemann bilinear relation, stating that for any two closed differentials $\gamma_1$ and $\gamma_2$,
\be\label{r1}
\int_\Sigma\gamma_1\wedge\gamma_2=\sum_{j=1}^\g\;\int_{a_j}\gamma_1\cdot\int_{b_j}\gamma_2-\int_{b_j}\gamma_1\cdot\int_{a_j}\gamma_2,
\ee 
see e.g.\ \cite[III.3]{FKra} for the proof.

The harmonic one-forms are related to the holomorphic differentials by the following linear transformation
\begin{align}\label{alpha1}
\alpha=\frac i2\bar\tau(\Im\tau)^{-1}\omega-\frac i2\tau(\Im\tau)^{-1}\bar\omega,\quad
\beta=-\frac i2(\Im\tau)^{-1}\omega+\frac i2(\Im\tau)^{-1}\bar\omega,
\end{align}
where $(\Im\tau)^{-1}$ is the inverse matrix and the matrix/vector summation notation is understood from now on. 

Holomorphic differentials on the Riemann surface can be used to define the Abel map. The period matrix of $\Sigma$ generates a lattice 
$\Lambda\in\mathbb C^{\g}$,
$$
\Lambda=\{m'_j+\tau_{jl}m_l\,|\;m,m'\in\mathbb Z^\g\}.
$$
The Jacobian variety or Jacobian of the Riemann surface is the complex torus 
\begin{equation}\nonumber
Jac(\Sigma)=\mathbb C^\g/\Lambda.
\end{equation}
We define the Abel map $I[P],\,P\in\Sigma$ with respect to the base point $P_0\in\Sigma$ as the holomorphic map
\begin{align}\label{abel}\nonumber
I:\;&\Sigma\to Jac(\Sigma),\\
&P\to\left(\int_{P_0}^P\omega_1,...,\int_{P_0}^P\omega_\g\right)\mod\;\Lambda
\end{align}
This is well-defined, i.e. independent of the choice of the integration paths, since the period integrals in \eq{acycl} belong to the lattice $\Lambda$. The Abel map is a holomorphic embedding of the Riemann surface into $\g$-dimensional complex torus $\mathbb C^\g/\Lambda$. Let us note, that the Abel map is also well-defined on the divisors, i.e. on the formal sums of points with integer multiplicities 
\be\nonumber
D=\sum_a n_a\cdot P_a,\quad I[D]=\sum_a n_a\cdot I[P_a].
\ee
The degree of a divisor $D$ is defined as the sum of the multiplicities of points,
$$\deg D:=\sum_an_a.$$ 

Meromorphic functions on $\Sigma$ as well as holomorphic sections of line bundles, which will be defined later, can be characterized by the divisors of their poles and zeroes taken with multiplicities. Abel theorem \cite[vol.\ I, Prop.\ 2.4]{M} states that a divisor $D$ is a divisor of a meromorphic function $f$ on $\Sigma$ if and only if $\deg D=0$ and its image under the Abel map vanishes $I[D]=0\mod\Lambda$. Divisors of meromorphic functions are called principal and the notation is $D=(f)$.

Taking into account the Abel theorem it is customary to introduce the equivalence classes on the space of all divisors, where $D\equiv D'$ iff there exists a meromorphic function such that $D=D'+(f)$. In this case $D$ and $D'$ belong to the same divisor class $[D]$. We are mostly interested in positive divisors $D>0$  (also called integral or effective), i.e. those with all multiplicities positive $n_a>0$. Hence, we define the set $J_d$ 
\be\nonumber
J_d:=\{[D]\;|\; D>0,\deg D=d\},
\ee
of equivalence classes of positive divisors of degree $d$.

\subsection{Theta functions and theta divisor}
\label{thet}

Consider a coordinate vector $e\in\mathbb C^\g$ and the lattice $\Lambda=m'_j+\tau_{jl}m_l,\;m,m'\in\mathbb Z^\g$, where $\tau$ is symmetric complex matrix with positive definite imaginary part ${\rm Im}\,\tau>0$. Such matrices are parameterized by an open subset in $\mathbb C^{g(g+1)/2}$, called the Siegel upper-half plane.

Theta function is an analytic function of $e\in\mathbb C^\g$ defined by the series
\be\nonumber
\vartheta(e,\tau)=\sum_{n\in\mathbb Z^\g}e^{\pi i n\tau n+2\pi ine},
\ee
where $n \tau n:=\sum_{j,l}n_j\tau_{jl}n_l$ and $ne:=\sum_jn_je_j$.
This function is quasi-periodic with respect to the lattice shifts
\be\nonumber
\vartheta(e+m'+\tau m,\tau)=e^{-\pi i m\tau m-2\pi ime}\vartheta(e,\tau),\quad m,m'\in\mathbb Z^\g.
\ee
Theta function with characteristics is defined as follows
\be\label{deftheta}
\vartheta\Big[{\tiny\begin{array}{c}a\\b\end{array}}\Big](e,\tau)
=e^{\pi ia\tau a+2\pi ia(e+b)}\vartheta(e+a\tau+b,\tau),\quad a,b\in\mathbb R^\g,
\ee
and its quasi-periodicity properties are given by
\begin{equation}\label{qp1}
\vartheta\Big[{\tiny\begin{array}{c}a\\b\end{array}}\Big](e+m'+\tau m,\tau)=e^{-\pi i m\tau m-2\pi ime+2\pi i(am'-bm)}\vartheta\Big[{\tiny\begin{array}{c}a\\b\end{array}}\Big](e,\tau),
\end{equation}

Theta functions are naturally associated with the Jacobian variety of a Riemann surface upon identification of $\tau$ with the period matrix \eq{acycl}. Moreover, the zero set of $\vartheta(e)$ in $\mathbb C^\g$, called the theta divisor, can be described in terms of the set $J_{\g-1}$ of positive divisors of degree $\g-1$ on $\Sigma$. 

\begin{rem}
In what follows we often suppress the period matrix from the notation for the theta function and shorten it to $\vartheta(z)$, unless the full notation is necessary. 
\end{rem} 

\begin{prop}\label{thetadivisor}\cite[\S VI.3.1]{FKra}
Let $e\in \mathbb C^\g$. Then $\vartheta(e)=0$ iff there exists a divisor $D\in J_{\g-1}$ satisfying
$$
e=I[D]+\Delta.
$$
\end{prop}
The vector $\Delta$ is called the vector of Riemann constants. It has the following form
\be\label{riemannC}
\Delta_j=\frac12+\frac12\tau_{jj}-\sum_{\substack{l=1\\l\neq j}}^\g\int_{a_l}\left(\omega_l(z)\int_{P_0}^z\omega_j\right),
\ee
and it depends on the choice of the point $P_0\in\Sigma$, which will sometimes be indicated explicitly by the superscript $\Delta^{P_0}$. The following transformation formula holds
\be\label{riemannCz}
\Delta^z=\Delta^{P_0}+(\g-1)\int_{P_0}^z\omega.
\ee
Theta functions can be also used to explicitly construct meromorphic functions and holomorphic sections on the Riemann surface $\Sigma$. 
Consider the following theta function
\begin{equation}
\label{fP}
\vartheta\left(\int_{P_0}^z\omega-e\right)
\end{equation}
as a function of a point $z\in \Sigma$, for an arbitrary fixed vector $e\in \mathbb C^\g$ and fixed base point $P_0\in\Sigma$. 
This function is well-defined locally on $\Sigma$, but multi-valued globally. 
It remains invariant around the $a$-cycles and transforms multiplicatively under the continuation around the $b_j$-cycle,
\begin{equation}\nonumber
\vartheta\left(\int_{P_0}^z\omega-e+\int_{b_j}\omega\right)=e^{-\pi i\tau_{jj}-2\pi i\left(\int_{P_0}^z\omega_j-e_j\right)}\,\vartheta\left(\int_{P_0}^z\omega-e\right),
\end{equation}
for a fixed index $j$.
Since the multiplicative factor is non-vanishing the zeroes of the function \eqref{fP} are well-defined on $\Sigma$. They are described by the Riemann vanishing theorem as follows.
\begin{prop}\label{RV} \cite[vol.\ I, Thm.\ 3.1]{M}
Theta function \eq{fP} either vanishes identically on $\Sigma$ or has $\g$ zeroes (counting with multiplicities) $P_1,...,P_\g$. In the latter case there exists a vector $\Delta\in\mathbb C^\g$, \eq{riemannC}, such that
\be\label{riemvan}
\sum_{j=1}^\g\int_{P_0}^{P_j}\omega=e-\Delta\;\mod\Lambda.
\ee
\end{prop}

Now we take a point $e\in\mathbb C^\g$ such that $\vartheta(e)=0$, and consider \eqref{fP} 
as a function on $\Sigma\times\Sigma$,
\be\label{e}
f_{e}(z,y)=\vartheta\left(\int_z^y\omega-e\right),
\ee
The following description of zeroes in $z$ and $y$ of this function is a consequence of the Riemann vanishing theorem. 
\begin{prop}\label{van}\cite[vol. I, Lem.\ 3.4]{M}
Let $e\in \mathbb C^\g$ such that $\vartheta(e)=0$ and $f_{e}(z,y)\not\equiv0$. 
Then there are $2\g-2$ points (counting with multiplicities) $P_1,...,P_{\g-1},Q_1,...,Q_{\g-1}\in \Sigma$ such that
$f_e(z,y)=0\iff (z=y),\;{\rm or}\; z=P_j, \forall y\in \Sigma\;{\rm or}\; y=Q_j, \forall z\in \Sigma$.
\end{prop}
The function $f_e(z,y)$ will be used in the construction of the Prime form in \S\ref{primeform}. Before defining the Prime form we introduce theta characteristics.

\subsection{Theta characteristics and spinors}
\label{sstr}

Consider half-periods of the period lattice $\Lambda$
\be\nonumber
\delta=\delta'+\delta'' \tau, \quad \delta',\delta''\in \big\{0,{\textstyle\frac12}\big\}^\g,
\ee
These are called half-integer characteristics or theta characteristics. It follows from the definition \eq{deftheta} that under the reflection $e\to-e$ in $\mathbb C^\g$ the theta function transforms as
\be\label{oddtheta}
\vartheta\Big[{\tiny\begin{array}{c}\delta'\\\delta''\end{array}}\Big](e)=e^{-4\pi i\delta'\delta''}\vartheta\Big[{\tiny\begin{array}{c}\delta'\\\delta''\end{array}}\Big](-e),\quad {\rm where}\;\;\delta'\delta'':=\sum_{j=1}^\g\delta'_j\delta''_j
\ee
Depending on parity of $4\delta'\delta''$ the theta characteristic is called odd or even. Among $4^\g$ different theta characteristics there are $2^{\g-1}(2^\g+1)$ even and $2^{\g-1}(2^\g-1)$ odd ones. It follows from \eq{oddtheta} that $\vartheta[\delta](0)=0$ for all odd theta characteristics. 

By analogy with the meromorphic functions the divisor $(\eta)$ of an Abelian (meromorphic) differential $\eta(z)$ can also be defined as a formal sum of its poles and zeroes with multiplicities. Writing $\eta=f(z)dz$ in terms of a local holomorphic coordinate the orders of poles and zeroes are well defined and the divisor is a sum of distinct poles and zeroes of $f(z)$ over all coordinate charts. 
Since the ratio of any two Abelian differentials is a meromorphic function, by Abel theorem all Abelian differentials belong to the same divisor class $C$, called the canonical class. The image of the canonical divisor under the Abel map satisfies the following relation, see \cite[\S VI.3.6]{FKra},
$$
I[C]=-2\Delta,
$$
where $\Delta$ is the vector of Riemann constants \eqref{riemannC}.  

An immediate consequence of this result and Prop.\ \ref{thetadivisor} is that
there exists a positive divisor $D_\delta$ of degree $\g-1$,
\be\label{Ddelta}
D_\delta=\sum_{\alpha=1}^{\g-1}p_\alpha
\ee
corresponding to each {\it odd} theta characteristic, satisfying $2D_\delta\equiv C$ and
\be\label{del}
\delta=I[D_\delta]+\Delta.
\ee 
It follows that there exists holomorphic differential $\omega_\delta$ with the divisor $(\omega_\delta)=2D_\delta$, i.e., with the double zeroes exactly on $D_\delta$. This differential can be constructed using the theta functions in the case of so called non-singular theta characteristic.

Theta characteristic $\delta_*$ is called non-singular, if there exists non-vanishing partial derivative of $\vartheta$ at $\delta_*$, $\frac{\p\vartheta}{\p e_j}(\delta_*)\neq0$, at least for some $j$'s. Existence of non-singular odd theta characteristic follows from the Lefschetz embedding theorem, see e.g.\ \cite[vol.\ II, p.\ 3.208]{M}. Let $D_{\delta_*}$ is the divisor \eq{Ddelta}, corresponding to $\delta_*$. Then the holomorphic differential $\omega_{\delta_*}$, with the divisor $(\omega_{\delta_*})=2D_{\delta_*}$ is given by
\be\label{omegadelta}
\omega_{\delta_*}(z)=\sum_{j=1}^\g\frac{\p\vartheta[\delta_*]}{\p e_j}\;\omega_j(z),
\ee
where $\omega_j$ is the canonical basis of holomorphic differentials. 

The half-differential or spinor, is a holomorphic object $\sqrt{\omega_{\delta_*}}$ with the divisor $D_{\delta_*}$ and such that $\sqrt{\omega_{\delta_*}}^2=\omega_{\delta_*}$. This is a holomorphic section of the line bundle $S$ of the degree $\g-1$, satisfying $S^2=K$, where $K$ is the canonical line bundle, i.e., holomorphic line bundle of Abelian differentials. The canonical line bundle corresponds to the canonical divisor class $C$. Given a half-integer characteristic $\varepsilon=(\varepsilon',\varepsilon'')$, the analog of the relation \eqref{del} holds
\be\label{spineps}
\varepsilon=I[D_\varepsilon]+\Delta,
\ee 
where $D_\varepsilon$ is (not generically positive) divisor of degree $\g-1$, defining the spin bundle $S_\varepsilon$ with spin structure $(\varepsilon',\varepsilon'')$. There are $4^\g$ non-isomorphic spin bundles on the Riemann surface of genus $\g$, by the number of half-integer characteristics.

\subsection{Prime form}
\label{primeform}

The idea is to construct a holomorphic function $E(z,y)$ on $\Sigma\times\Sigma$ with a unique simple zero on the diagonal $z=y$. This is an analog of the linear function $z-y$ on $\mathbb P^1\times\mathbb P^1$ for a genus $\g>1$ Riemann surface. 

The function $f_e(z,y)$ \eq{e} vanishes at $z=y$, but it has extra $\g-1$ zeroes according to Prop.\ \ref{van}. Take a non-singular theta characteristic $\delta_*$ and set $e=\delta_*$. Then it follows from \eq{riemvan} and \eq{del} that in this case the extra $\g-1$ zeroes belong precisely to the divisor $D_\delta$. Hence they can be cancelled by dividing by $\sqrt{\omega_{\delta_*}}$ in each variable.

The resulting object
\be\label{pform}
E(z,y)=\frac{\vartheta\Big[{\tiny\begin{array}{c}\delta_*'\\\delta_*''\end{array}}\Big]\Big(\int_z^y\omega\Big)}{\sqrt{\omega_{\delta_*}(z)}\sqrt{\omega_{\delta_*}(y)}},
\ee
is called the Prime form, see e.g., \cite[vol.\ II, sec.\ 3b, \S 1]{M}. Since \eqref{pform} depends on the choice of the path of integration, $E(z,y)$ is well-defined only on the universal cover of $\Sigma$. However the zeros of the Prime form is well-defined on $\Sigma$ and $E(z,y)$ has only one first order zero at $z=y$. It also transforms as a holomorphic $(-\frac12)$-differential in each variable. 
Other useful properties of the Prime form include antisymmetry
$$E(z,y)=-E(y,z),$$ 
which follows from \eq{oddtheta} since the theta characteristic $\delta$ is odd. 
The near-diagonal behavior is given by
\be\label{neard}
E(z,y)=\frac{z-y}{\sqrt{dz}\sqrt{dy}}(1+\mathcal O((z-y)^2))).
\ee
The Prime form is also independent of the choice of the odd non-singular characteristic. 
This follows from the fact, that the periodicity properties of the Prime form along $a,b$-cycles are independent of the choice of $\delta_*$. Indeed, it is invariant under continuation around the $a$-cycles and around $b_j$-cycle it transforms as
\begin{align}\label{transE}
&E(z+a_j,y)=\left({\textstyle\frac{dz'}{dz}|_{a_j}}\right)^{-\frac12}E(z,y),\\\label{transE1}
&E(z+b_j,y)=\pm\left({\textstyle\frac{dz'}{dz}|_{b_j}}\right)^{-\frac12}e^{-\pi i\tau_{jj}-2\pi i\int_y^z\omega_j}E(z,y),
\end{align}
and similarly for $y$. Here ${\textstyle\frac{dz'}{dz}|_{a_j}}$, ${\textstyle\frac{dz'}{dz}|_{b_j}}$ are the automorphy factors for the Abelian differentials, i.e., for the canonical line bundle $K$.

Finally, the Prime form can be used to construct all meromorphic functions on $\Sigma$. If $(f)=\sum n_a\cdot P_a$ is a degree zero divisor of a meromorphic function $f(z)$, then the function has the form 
\be\nonumber
f(z)=C\cdot e^{2\pi i\int_{}^z\omega}\cdot \prod_a E(z,P_a)^{n_a},
\ee
up to multiplication by a constant $C$ and by a trivial automorphy factor, cancelling out the $b$-cycle transformations in \eqref{transE1}.
Using the behavior under the continuation along the $b$-cycles \eq{transE} one can check that $f$ is indeed a globally defined function, provided the image of the divisor under the Abel map is zero $I[(f)]=0$, in agreement with the Abel theorem.

\section{Lowest Landau level, integer QH state and the Laughlin state}

\subsection{Holomorphic line bundles on Riemann surfaces}

The lowest Landau level (LLL) on a Riemann surface is defined as the space of global solutions to the $\bp$-equation
\begin{equation}
\label{dbareq}
\bp_{\rm L} s(z)=0.
\end{equation}
These are the holomorphic sections of the line bundle $\rm L$. Here $\bp_{\rm L}$ is the Dolbeault operator, which 
acts from $\mathcal C^\infty$ sections of $\rm L$ to (0,1) forms with coefficients in $\mathcal C^\infty$ sections,
\begin{equation}\nonumber
\bp_{\rm L}: \mathcal C^\infty(\Sigma,\rm L)
\to\Omega^{0,1}(\Sigma,\rm L).
\end{equation}
In the context of quantum Hall states, see e.g., \cite{K16}, we choose the line bundle as a tensor product 
\be\label{rmL}
\L=L_\NPhi\otimes K^s\otimes L_{\varphi}.
\ee
Here $L_\NPhi$ is a positive line bundle of degree $\NPhi$ (here $\NPhi$ in general denotes the degree and not a tensor power), $K^s$ is the power of the canonical line bundle, $s\in\frac12\mathbb Z$, and $L_{\varphi}$ is a flat line bundle.
The total degree of the line bundle $\L$ is 
\be\label{degL}
\deg\L=\NPhi+2s(\g-1).
\ee
For the half-integer spin $s$, $K^s\simeq K^{[s]}\otimes K^{1/2}$ and we shall specify the choice of the square root of the canonical bundle. In this case instead of $K^s$ we shall write
$$
S_\delta^{2s}\simeq K^{[s]}\otimes S_\delta,
$$
where $[s]$ is the integer part of $s$ and $S_\delta$ is the spin bundle, labelled by the spin structure $\delta$, see \eq{spineps}. Then the holomorphic sections \eq{dbareq} correspond to the LLL states on $\Sigma$ for the magnetic field with total flux $\NPhi$, with the AB-phases encoded by $L_{\varphi}$ and the gravitational spin $s$. The latter also means that the sections of the line bundle \eq{rmL} transform as an integer or half-integer tensor power of a holomorphic one form $(dz)^s$.

In the context of the higher genus Riemann surfaces it is convenient to describe the  
holomorphic sections of line bundles in terms of their divisors. Here we recall the basic construction and refer to \cite{Bost,Fay92} for general review and explicit examples.

Consider an arbitrary positive holomorphic line bundle $L_d$ of degree $d$. Suppose we have constructed a holomorphic section $s(z)$ of $L_d$ with a positive divisor $D$ of zeroes and no poles, 
\begin{equation}\nonumber
D=\sum_{a=1}^d q_a
\end{equation}
where we allow for coincident points. The corresponding holomorphic line bundle is also denoted as $L_d\simeq \mathcal O(D)$. The Riemann-Roch theorem \cite[III.4.8]{FKra} computes the dimension of the vector space 
\be\label{HO}
H^0(\Sigma,\mathcal O(D)):=\{\,f\;|\;(f)\geqslant - D\}
\ee
of meromorphic functions $f$ with divisors of poles contained in $D$. We have
\be\label{RRoch}
\dim H^0(\Sigma,\mathcal O(D))=\deg D-\g+1+\dim H^0(\Sigma,\Omega(-D)),
\ee
where 
\be\nonumber
H^0(\Sigma,\Omega(-D))\}:=\{\,\eta\;|\;(\eta)\geqslant D\},
\ee
is the vector space of holomorphic differentials $\eta$ with divisor of zeroes that includes $D$. Dimension of this vector space is also called the index of speciality of $D$, denoted as
\be\label{speciality}
i(D):=\dim H^0(\Sigma,\Omega(-D)).
\ee
We will consider the case of $\deg D$ large enough, so that no such differentials exist and the index of speciality is identically zero. The lower bound on $\deg D$ for this to be the case 
is $\deg D>2\g-2$, as follows immediately from the definition of $H^0(\Sigma,\Omega(-D))\}$ taking into account that $\deg (\eta)=2\g-2$. Therefore in the case of the line bundle \eq{rmL} the correction term in \eq{RRoch} vanishes for
\be\label{bound}
\NPhi>(1-s)(2\g-2).
\ee
Now, taking any meromorphic function $f(z)$ from $H^0(\Sigma,\mathcal O(D))$, the product $f(z)\cdot s(z)$ is again a holomorphic section of $\mathcal O(D)$ since the poles of $f(z)$ cancel out by the zeroes of $s(z)$. Hence the vector space $H^0(\Sigma,L_d)$ of the holomorphic sections of $L_d\simeq \mathcal O(D)$ has the dimension 
\be\nonumber
\dim H^0(\Sigma,L_d)=d-\g+1,
\ee
where we set the last term in \eq{RRoch} to zero. 

Now we apply this logic to the line bundle $O(D_\m)\otimes S^{2s}_\delta$. We denote the magnetic field divisor corresponding to the line bundle $L_\NPhi\simeq\mathcal O(D_\m)$ as
\begin{equation}\label{degDm}
D_\m=\sum_{a=1}^\NPhi q_a,
\end{equation}
and 
\be\label{degDd}
D_\delta=\sum_{\alpha=1}^{\g-1} p_\alpha,\quad \deg D_\delta=\g-1
\ee
is the divisor of the spin bundle $S_\delta$, labelled by the spin structure $\delta$, see  \eq{spineps}. For simplicity we restrict our discussion here to the case of an arbitrary odd (and not necessarily non-singular) theta characteristic $\delta$, where the divisor $D_\delta$ is positive, \eq{Ddelta}.

Now we would like to construct the holomorphic sections of the line bundle $\L$ explicitly.

\subsection{Lowest Landau level}

As follows from the discussion above, in order to describe the full lowest Landau level $H^0(\Sigma,\L)$ we can construct a single holomorphic section of $\L$ and then apply Riemann-Roch theorem to compute the degeneracy $\dim H^0(\Sigma,\L)$. 
Consider first the magnetic line bundle $\mathcal O(D_\m)$.
Using the Prime form, \eq{pform}, we can immediately construct the following object
\begin{equation}\label{prodE}
\prod_{a=1}^\NPhi E(z,q_a),
\end{equation}
with the divisor of zeroes $D_\m$. However, this object transforms as $(-\NPhi/2)$-differential in $z$ and not as a scalar. This can be remedied by introducing the following $\frac\g2$-differential with no zeroes and poles,
\be\label{sigmafunc}
\sigma(z)=\exp-\sum_{j=1}^\g\int_{a_j}\omega_j(y)\log E(y,z).
\ee
Following \cite[Prop.\ 1.2]{Fay92}, $\sigma(z)$ is characterized by the automorphy factors
\be\label{autosig}
\left({\textstyle\frac{dz'}{dz}|_{a_j}}\right)^{\frac\g2},\quad \left({\textstyle\frac{dz'}{dz}|_{b_j}}\right)^{\frac\g2}\cdot e^{\pi i(\g-1)\tau_{jj}+2\pi i\Delta_j^z},
\ee
under the continuation along the $a$ and $b$-cycles, where 
${\textstyle\frac{dz'}{dz}|_{a_j}}$, ${\textstyle\frac{dz'}{dz}|_{b_j}}$ are the automorphy factors \eq{transE} and $\Delta^z$ is the vector of Riemann constants defined with respect to the point $z$, as in \eq{riemannCz}.

\begin{rem}
We shall note that $\sigma(z)$ is properly defined only on the universal cover of $\Sigma$ or on the 2-cell for the canonical dissection of the surface, introduced later in Fig.\ \ref{fig:dissect}, and does not lift to the whole Riemann surface as a holomorphic section of any line bundle. We also note, that instead of using the definition above for $\sigma(z)$ we can equivalently use the $-\g(\g-1)/2$-differential $c(z)$ defined in Ref. \cite[Eq.\ 1.17]{Fay92}, where an explicit formula is given in terms of holomorphic differentials.
\end{rem}

Multiplying \eqref{prodE} by the power of $\g$-th root of $\sigma$ we can construct the scalar
$$
\sigma(z)^{\NPhi/\g}\cdot\prod_{a=1}^\NPhi E(z,q_a)
$$
which is now a holomorphic section of $\mathcal O(D_\m)$ on $\Sigma$, considered as a function of $z$. 
Next we take the flat line bundle in \eq{rmL} into consideration and construct the sections of $\mathcal O(D_\m)\otimes L_\varphi$. The idea here is to parameterize all line bundles of degree $\NPhi$.    
The space of equivalence classes of line bundles $L_d$ of fixed degree $d$ is called Picard variety
\begin{equation}\nonumber
{\rm Pic}_d(\Sigma)=\{{\rm line\; bundles}\; L\; {\rm on}\; \Sigma\; {\rm with}\; c_1(L)=d\}.
\end{equation}
Picard varieties for different $d$ are very similar and can be identified with ${\rm Pic}_0$, which in turn can be identified with the Jacobian $Jac(\Sigma)$, in the following fashion. For instance we can start with a choice of a divisor \eq{degDm}. Then for any other inequivalent divisor $D'_\m=\sum_{a=1}^\NPhi q'_a$ in ${\rm Pic}_d(\Sigma)$ the holomorphic sections of $\mathcal O(D_\m)$ and $\mathcal O(D'_\m)$ can be identified by multiplication by a meromorphic section 
\be\label{merosec}
\prod_{a=1}^\NPhi\frac{E(z,q'_a)}{E(z,q_a)},
\ee
of the degree zero line bundle with the divisor $D'_\m-D_\m$.
Using the Abel map we can parameterise all inequivalent divisors of the same degree
by a complex vector
\be\label{varphi}
I[D'_\m-D_\m]=\varphi \in Jac(\Sigma),
\ee
taking values in the Jacobian of the Riemann surface. If $D_\m\equiv D'_\m$ then \eq{merosec} defines a meromorphic function and in this case $\varphi=0 \mod\Lambda$ by the Abel theorem.
In other words we can represent $\mathcal O(D'_\m)$ as $\mathcal O(D_\m)\otimes L_\varphi$ using the relative complex coordinate $\varphi$ in $Jac(\Sigma)$.
Reversing the logic, we can describe the line bundle $\L$ \eq{rmL} by the divisor 
\be\label{divL}
{\rm div}\,\L=D_\m+2sD_\delta+D'_\m-D_m,
\ee
where the image of $D'_\m-D_m$ in the Jacobian is given by $\varphi$, \eq{varphi}.

\begin{rem}
There is no canonical way to fix the coordinates on ${\rm Pic}_d$ universally for all $d$. A particularly convenient choice is to fix one marked point $P\in \Sigma$ and consider the divisor
\begin{equation}\nonumber
D^0_\m=\NPhi\cdot P,
\end{equation}
concentrated at $P$ as a reference point in ${\rm Pic}_d$. Then the corresponding line bundle $\mathcal O(D^0_\m)\simeq L^\NPhi$ is actually the $\NPhi$th tensor power of degree one line bundle $\mathcal O(P)$. We can use this divisor as a base point for the magnetic line bundles, so that $I[D'-D^0_\m]=I[D'-D]+I[D-D^0_\m]=\varphi+\varphi_0$. Although we do not pursue this further, this choice is useful, e.g., for the consistent definition of large $\NPhi$ asymptotics. 
\end{rem}

Putting everything together we can immediately write down the following holomorphic section $s(z)$,
\begin{equation}\label{LLLs}
s(z)=\sigma(z)^{2s+\NPhi/\g}\cdot\prod_{a=1}^\NPhi E(z,q_a)\cdot\prod_{\alpha=1}^{\g-1}\left(E(z,p_\alpha)\right)^{2s}\cdot\prod_{a=1}^\NPhi\frac{E(z,q'_a)}{E(z,q_a)}
\end{equation}
of $\L$, \eq{rmL}, with the positive divisor $(s)= D_\m+2sD_\delta+D'_\m-D_\m$.

By the Riemann-Roch theorem \eqref{RRoch} the dimension of the vector space of holomorphic sections, i.e., the degeneracy of the LLL states is
\be\label{NNPhi}
N=\dim H^0(\Sigma,\L)=\NPhi+(2s-1)(\g-1),
\ee
provided the lower bound \eq{bound} on $\NPhi$ holds. In what follows it will be enough for our purposes to construct just one section and we will not need to construct the full basis explicitly.

An equivalent way to describe the lowest Landau level is via the automorphy factors, i.e., the multipliers $\chi_a$ and $\chi_b$ appearing under the transformations of the holomorphic sections along $a$ and $b$ cycles. Indeed, since the ratio of two holomorphic sections of the same bundle is a meromorphic function on $\Sigma$ all sections in $H^0(\Sigma,\L)$ have the same automorphy factors.
We summarize this observation in the following
\begin{prop}\label{auto}
The line bundle $\L=\mathcal O(D_\m)\otimes S_\delta^{2s}\otimes L_\varphi$ is described by the following automorphy factors
\be\label{auto1}
\chi_{a_j}=\left({\textstyle\frac{dz'}{dz}|_{a_j}}\right)^s,\quad \chi_{b_j}=\left({\textstyle\frac{dz'}{dz}|_{b_j}}\right)^s\cdot e^{-\pi i\frac\NPhi\g\tau_{jj}+2\pi i(2s\delta_j+\varphi_j)+2\pi iI_j\,[D_\m-\NPhi z]+2\pi i\frac\NPhi\g\Delta_j^z}.
\ee 
Here ${\textstyle\frac{dz'}{dz}|_{a_j}}$, ${\textstyle\frac{dz'}{dz}|_{b_j}}$ are the automorphy factors \eq{transE} for the canonical line bundle $K$ and $I_j[\,\cdot\,]$ denotes the $j$th component of the Abel map \eq{abel}.
\end{prop}
\begin{proof}
Follows from the transformation formulas for the Prime form, \eq{transE}, and for $\sigma(z)$, \eq{autosig}, under continuation along the $a$, $b$-cycles, taking into account the relation \eqref{del}.
\end{proof}

\subsection{Integer QH state}
\label{intQH}

Suppose $s_n(z),\;n=1,...,N$, where $N=\dim H^0(\Sigma,\L)$, is a full basis of the states on the lowest Landau level. 
\begin{maindefn}\label{IQHS}
The integer QH state is a section of $\det \oplus_n\pi^*_n(\L)$ on $\Sigma^N$, where $\pi_n$ is the restriction to the $n$th factor in $\Sigma^N$. Explicitly, 
\be\nonumber
\det s_n(z_m)|_{n,m=1}^N.
\ee
In other words this is a completely anti-symmetric combination of one-particle states, i.e., the Slater determinant. 
\end{maindefn}
Although we have constructed only one holomorphic section explicitly \eq{LLLs}, the integer QH state nevertheless can be constructed explicitly in terms of theta functions and Prime forms. The following is a version of the Fay's identity \cite[Prop.\ 2.16]{Fay73}
\begin{prop}
Let $s_n(z),\;n=1,...,N$, $N=\dim H^0(\Sigma,\L)$, be a basis of holomorphic sections of the line bundle $\L=\mathcal O(D_\m)\otimes S^{2s}_\delta\otimes L_\varphi$, \eq{rmL}, such that its index of speciality \eq{speciality} vanishes $i({\rm div}\,\L)=0$. Then for $z_1,...,z_N\in \Sigma^N$ the following formula holds
\be\label{boso}
\det s_n(z_m)|_{n,m=1}^N=C(\Sigma)\cdot\vartheta\Big({\textstyle \sum_{n=1}^N} z_n-\Delta-{\rm div}\,\L,\tau\Big)\cdot\prod_{n<m}^NE(z_n,z_m)\cdot\prod_{n=1}^N\sigma(z_n)^{\frac1\g(N+2s-1)},
\ee
where $\Delta$ is the vector of Riemann constants \eq{riemannC} and $C(\Sigma)$ is a $z$-independent constant, which can depend on the surface $\Sigma$.
\end{prop}
\begin{rem}
We shorten the notation and drop $I[\,\cdot\,]$, where the context implies the Abel map, as e.g. in the notations for the theta functions. In particular, taking into account 
\eq{divL} and \eq{varphi}, it follows that
$$
I[\,{\rm div}\,\L\,]=I[D_\m+2sD_\delta]+\varphi
$$
in the argument of the theta function above, \eq{boso}.
\end{rem}
\begin{proof}
Since both sides are manifestly holomorphic with no poles and also completely anti-symmetric in $z_n$'s it is enough to compare the zeroes in $z_1$. The lhs is a section of $\L$ with the divisor
\be\nonumber
D_{\rm lhs}=\sum_{m=2}^Nz_m+P_1+...+P_\g,
\ee
for certain points $P_j$ and by definition $D_{\rm lhs}\equiv 2sD_\delta+D_\m+\varphi$. (Here with some abuse of notations, $D_\m+\varphi$ stands for $D'_\m$, see \eq{divL}, \eq{varphi}) By description of the theta divisor in Prop.\ \ref{thetadivisor} the theta function on the rhs would vanish identically iff\; $2sD_\delta+D_\m+\varphi-\sum_{n=1}^Nz_n$ is linearly equivalent to a positive divisor in $J_{\g-1}$. However, this cannot happen since it would imply $\dim\L>N$.
Hence by the Riemann vanishing theorem, Prop.\ \ref{RV}, the theta function on the rhs has exactly $\g$ zeroes at the points $\tilde P_j$. Hence the divisor in $z_1$ of the rhs is
\be\nonumber
D_{\rm rhs}=\sum_{m=2}^Nz_m+\tilde P_1+...+\tilde P_\g
\ee
and moreover by \eq{riemvan} we have
$$
\sum_{j=1}^\g\tilde P_j\equiv-\sum_{m=2}^Nz_m+2sD_\delta+D_\m+\varphi
$$
and it follows that $D_{\rm rhs}\equiv 2sD_\delta+D_\m+\varphi\equiv D_{\rm lhs}$. Now we need to show that the divisors $D_{\rm lhs}$ and $D_{\rm rhs}$ actually coincide, i.e. that the following degree $\g$ divisors coincide
$$
P_1+...+P_\g=\tilde P_1+...+\tilde P_\g,
$$
and what we have established so far is only their linear equivalence $\sum P\equiv\sum \tilde P$. Suppose there exists a non-constant meromorphic function $f$, such that 
\be\label{PP}
P_1+...+P_\g=\tilde P_1+...+\tilde P_\g+(f).
\ee
Then by \eq{RRoch} it follows that $\dim H^0(\Sigma,\mathcal O(\tilde P_1+...+\tilde P_\g))=\g-\g+1+i(\tilde P_1+...+\tilde P_\g)>1$ and thus the index of speciality \eqref{speciality}, $i(\tilde P_1+...+\tilde P_\g)>0$. By \cite[VI.3.3., Thm.\ b]{FKra} this would imply that the theta function on the rhs vanishes identically and we arrive at a contradiction. Hence $f=const$ and we have established the equality \eq{PP}.

Finally, since all $s_n(z)$ on the lhs have the same automorphy factors, described in Prop.\ \ref{auto}, one can immediately check that the automorphy factors on both sides of \eq{boso} coincide.

\end{proof}

\subsection{Hermitian metric and bosonization formula}

Hermitian metric $h(z,\bz)$ defines the length of a section $s(z)$ at a point $z\in\Sigma$  so that $\Vert s(z)\Vert_h^2$ is a scalar function on $\Sigma$. Since our line bundle is a tensor product, the hermitian metric is also a product
\be\label{norm1}
\Vert s(z)\Vert^2_h:=|s(z)|^2\cdot h^\NPhi(z,\bz)\big(g_{z\bz}(z,\bz)\big)^{-s},
\ee
Here $g_{z\bz}$ is the Riemannian metric on $\Sigma$ in conformal complex coordinates $ds^2=2g_{z\bz}d^2z$, and $(g_{z\bz})^{-s}$ transforms as $(-s,-s)$-differential, and thus can serve as a hermitian metric on $K^s$. Hermitian metric $h^\NPhi(z,\bz)$ on $\mathcal O(D_\m)\otimes L_\varphi$ defines the magnetic field, since the curvature $(1,1)$ form is equal to the magnetic field strength 
\begin{equation}\label{curvF}
F=-\p_z\p_{\bz}\log \bigl(h^{\NPhi}(z,\bz)\bigr)\,idz\wedge d\bz, \quad c_1(F):=\frac1{2\pi}\int_\Sigma F=\NPhi
\end{equation}
The magnetic field density is then defined as 
\be\nonumber
B=g^{z\bz}F_{z\bz},
\ee
and the quantization condition in \eq{curvF} reads
\be\label{quant}
\frac1{2\pi}\int_\Sigma B\sqrt gd^2z=\NPhi\in\mathbb Z_+.
\ee
In principle the magnetic field and the metric can be chosen in an arbitrary fashion, as long as all the necessary constraints are satisfied. However, there is convenient choice of the hermitian and Riemannian metrics where the bosonization formula has a particularly simple form. 
The choice of $F$ is the following
\be\label{Fzz}
F_{z\bz}=\NPhi\, {g_{\rm c}}_{z\bz}
\ee
where $g_c$ is the canonical metric on the surface
\be\label{canmet}
{g_{\rm c}}_{z\bz}\,idz\wedge d\bz=\frac{\pi i}{\g}\sum_{j,l=1}^\g\omega_j\wedge{(\Im \tau)^{-1}}_{jl}\bar\omega_l,
\ee
normalized to have the area $2\pi$.
The canonical metric is just the pullback of the flat metric on $Jac(\Sigma)$ to $\Sigma$ under the Abel map. The hermitian metric corresponding to \eq{Fzz} 
reads
\begin{multline}\label{here}
h^\NPhi(z,\bz)=\exp\left(-\frac{2\pi\NPhi}{\g(\g-1)^2}\Im \Delta^z\,(\Im\tau)^{-1}\Im \Delta^z\right.\\\left.+\frac{4\pi}{(\g-1)^2}\Im \Delta^z\,(\Im\tau)^{-1}\Im\big((\g-1)\,{\rm div}\, \L+\deg\L\cdot\Delta\big)\right).
\end{multline}
The divisor ${\rm div}\, \L$ and its degree are defined in \eq{divL} and \eq{degL}. Note that the last term in the parentheses in the second line is independent of $z$ as well as $P_0$.

This particular form of the hermitian metric is fixed by the conditions that, (1) \eq{Fzz} holds, (2) the norm of a section of $\L$, \eq{norm1}, is a scalar function, taking into account \eq{auto1} and (3) it is independent of the choice of the base point $P_0$.  

The bosonization formula establishes an equality for the hermitian norms of the objects in \eq{boso}, see Refs.\ \cite{Fal,ABMNV,Kn1986,VV},
\begin{align}\nonumber\label{bos1}
\frac{\det'\bp^+_\L\bp_\L^{\phantom{a}}}{\det\langle s_n,s_m\rangle_{L^2}}\,&\Vert\det s_n(z_m)\Vert_h^2\\&=
C\cdot\left(\frac{\det'\Delta_{g_a}}{A_{g_a}\det\Im\tau}\right)^{-\frac12}\cdot\Vert\vartheta\big({\textstyle\sum_{n=1}^Nz_n}-\Delta-{\rm div}\,\L,\tau\big)\Vert^2\cdot\prod_{n<m}^Ne^{-G_{g_a}^A(z_n,z_m)},
\end{align}
where $C$ is now a numerical constant, computed in Ref.\ \cite{Went}. In the numerator on the left hand side here stands the zeta determinant of the magnetic laplacian $\bp^+_\L\bp_\L^{\phantom{a}}$ for the line bundle $\L$, see e.g. \cite{KMMW} for the definition, and the point-wise hermitian norm is taken in the metric \eq{here}. In the denominator is the Gram matrix of $L^2$ inner products of sections. The objects on the rhs, such as the zeta determinant of the scalar laplacian $\det'\Delta_{g}$, the area $A_g$ and the Green function $G^A$, are computed in the Arakelov metric $g_a$ on $\Sigma$, which will be defined in \S \ref{arak}. The norm of the theta function is defined according to
\be\nonumber
\Vert\vartheta(e,\tau)\Vert^2=e^{-2\pi\Im e\,(\Im\tau)^{-1}\Im e}|\vartheta(e,\tau)|^2,
\ee
for any $e\in\mathbb C^\g$.
\subsection{Laughlin states}
Here we define and give an explicit form of the Laughlin state \cite{L} on a compact Riemann surface $\Sigma$ for the filling fraction
\be\nonumber
\nu=\frac1\beta
\ee
where $\beta=1,2,3,...$ is a positive integer. We again consider the line bundle  $\L$, \eq{rmL} and impose the condition that $\deg\L$ is divisible by $\beta$. Without loosing too much generality we can assume $\NPhi$ and $2s$ are both divisible by $\beta$,
\be\nonumber
\deg \L=\beta \left(\frac1\beta\NPhi+\frac{2s}\beta(\g-1)\right).
\ee
Note that we do not need the condition that the divisor $D_\m$ itself is a multiple of $\beta$. 

Now we consider the product $\Sigma^N$ with $N$ given by  
\be\label{NNp}
N=\frac1\beta\deg\L+1-\g=\frac1\beta\NPhi+\left(\frac{2s}\beta-1\right)(\g-1)
\ee
One can think of this number as corresponding to the number of LLL states for a line bundle  of degree $\deg\L/\beta$, by the Riemann-Roch formula \eqref{RRoch}. In other words, only a fraction of the all possible states in $H^0(\Sigma, \L)$ is filled.

\begin{maindefn}\label{Lsdef}
The Laughlin state $F(z_1,...,z_N)$ for the filling fraction $1/\beta$ on $\Sigma^N$ is defined as follows
\begin{enumerate}
\item Restriction $\pi_n F(z)$ to the $n$th component $\Sigma_n$ in $\Sigma^N$, i.e., $F(...,z_n,...)$ with all $z$'s but $z_n$ fixed, transforms as a holomorphic section of the line bundle $\L$, \eq{rmL}.

\item\label{2} For any pair $z_n$, $z_m$, $n\neq m$, the restriction to the diagonal $\pi_{nm}F(z_n,z_m)$ on $\Sigma_n\times\Sigma_m$ vanishes exactly to the order $\beta$ near the diagonal
\be\nonumber
\pi_{nm}F(z_n,z_m)\sim (z_n-z_m)^\beta,
\ee
in terms of a local complex coordinate.
\item\label{3} For $\beta$ odd $F$ is completely anti-symmetric and for $\beta$ even, completely symmetric on $\Sigma^N$.
\end{enumerate}
\end{maindefn}

\begin{rem}
Here $F(z_1,...,z_N)$ is the holomorphic part of the state, c.f. \eq{LaS}. The full wave function is properly $L^2$-normalized
$$
|\Psi(z_1,...,z_N)|^2=\frac{\Vert F(z_1,...,z_N)\Vert_h^2}{\langle\Vert F(z_1,...,z_N)\Vert_h^2\rangle_{L^2}},
$$
where $\Vert \cdot\Vert_h$ denotes the hermitian norm \eq{norm1} applied point-wise for each $z_n$ and $L^2$ norm on $\Sigma^N$ is induced from the $L^2$ norms on $\Sigma$ in obvious way.
\end{rem}

At $\beta=1$ this definition is actually equivalent to the Def.\ \ref{IQHS} of the integer QH state, due to the Fay's identity Prop.\ \ref{boso}, and thus the Laughlin state is unique. 
For a fixed $\beta>1$ there is more than one Laughlin states and they form a vector space, which we denote $\mathbb V_\beta$. The dimension of this vector space is conjectured to be $\dim \mathbb V_\beta=\beta^\g$ for genus-$\g$ Riemann surface \cite{WN}. Here we construct $\beta^\g$ linearly independent states in $\mathbb V_\beta$, provided the technical condition $N\geqslant\g$ holds. However we do not prove that this is a complete basis and $\dim \mathbb V_\beta=\beta^\g$. 

\begin{rem}
The gravitational spin $s$ is quantized in units of $\frac\beta2$, i.e., $s=0,\frac\beta2,\beta,$ etc. Other integer or half-integer values of $s$ can also be considered, at the expense of introducing quasi-holes \cite{L}. In other words, a completely filled $\nu=1/\beta$ ground state with the spin $s\neq\frac\beta2\mathbb Z$ will necessarily include at least one quasi-hole.
\end{rem}

\begin{prop}\label{indepst}
The following $\beta^\g$ Laughlin states are linearly independent 
\begin{align}\label{FL}\nonumber
&F_r(z_1,...,z_N)=\\
&=\vartheta\left[\begin{array}{c}\scriptstyle\frac r\beta\\\scriptstyle0\end{array}\right]\Big(\beta{\textstyle \sum_{n=1}^N} z_n-\beta\Delta-{\rm div}\,\L,\beta\tau\Big)
\cdot \prod_{n<m}^N{E(z_n,z_m)}^\beta\cdot\prod_{n=1}^N\sigma(z_n)^{\frac1\g(\beta N+2s-\beta)}
\end{align}
where $r\in(\mathbb N^+_{\leqslant\beta})^\g$ and we assume that $N\geqslant\g$.
\end{prop}
\begin{proof}
The conditions \ref{2} and \ref{3} of Def.\ \ref{Lsdef} are obviously satisfied, so we need to check that the rhs of \eq{FL} is a section of $\L$. Due to the (anti-)symmetry it is enough to check this for $z_1$. 

We note that the argument of the theta-function above does indeed span all of $Jac(\Sigma)$ for $N\geqslant\g$ by the the Jacobi inversion theorem \cite[III.6.6]{FKra}. Then the theta function $\theta[r/\beta,0](\beta (z-\Delta-\xi),\beta\tau)$ is the order-$\beta$ theta function on the Jacobian and by Riemann vanishing Prop.\ \ref{RV} its divisor is linearly equivalent to $\beta\xi$ for any $r=(1,...,\beta)^\g$. Hence the divisor of $F(z_1,...)$ \eq{FL} is linearly equivalent to
$$
-\beta\sum_{n=2}^Nz_n+{\rm div}\,\L+\beta\sum_{n=2}^N z_n={\rm div}\,\L.
$$
Using the lattice shift formulas for the Prime form \eq{transE} and the theta function \eq{qp1} we can then check that the rhs of \eq{FL} transforms as a section of $\L$ with the automorphy factors given by Prop.\ \ref{auto}.  
Finally, the dimension of the vector space of $\beta$th order theta functions in $\beta^\g$, see e.g. \cite[Prop.\ 1.3]{M}.
\end{proof}

In the next section we explain how the Laughlin states \eq{FL} can be generated from the path integral of a scalar field compactified on a circle. 


\section{Laughlin states from free fields}
\label{freef}

\subsection{Scalar field compactified on a circle}

The standard references for the gaussian path integrals on Riemann surfaces include \cite{ABMNV,VV,DP}, in this section we spell out the definition and the construction step by step. 

We write the conformal metric $g$ on $\Sigma$ in local complex coordinates as $d^2s=2g_{z\bz}dzd\bz$. Consider the real-valued free field $\sigma(z,\bz)$ on $\Sigma$, compactified on a circle of radius $R_c$, $\sigma\sim\sigma+2\pi R_c$, with the action functional
\begin{equation}
\label{action}
S(\sigma)=\frac1{2\pi}\int_\Sigma\bigl(\p_z\sigma\p_{\bz}\sigma+\frac{ i q}4\sigma R\sqrt g+i\sqrt\nu \sigma B\sqrt g\bigr)d^2z,
\end{equation}
where $R=-2g^{z\bz}\p_z\p_{\bz}\log g$ is the scalar curvature of the metric $g$, $\nu=1/\beta$ is the filling fraction and the real-valued constant $q$ is to be fixed later.
The scalar function $B$ is the magnetic field with the integer total flux over the surface equal to $\NPhi$, as in \eq{quant}.
Next, following \cite{MR}, we consider the string of vertex operators at the points $z_1,...,z_N$ and the expectation value given by the path integral
\begin{equation}\label{pathi}
\mathcal V\bigl(g,B,\{z_n\}\bigr)=\int e^{ i\sqrt\beta{\Large\sigma}(z_1,\bz_1)}\cdots e^{ i\sqrt\beta\sigma(z_N,\bz_N)}e^{-S(\sigma)}\mathcal D_g\sigma.
\end{equation}
The conformal dimension $s$ of the vertex operator $e^{i\sqrt{\beta}\sigma(z,\bz)}$ is given by
$
s=\frac12\sqrt\beta(\sqrt\beta-q)
$
and we can express $q$ in terms of $\sqrt\beta$ and $s$ as
$$
q=\sqrt\beta-\frac{2s}{\sqrt\beta}.
$$

Several constraints need to be imposed in order to make this setup well defined. First, 
the integration over the constant zero-mode of the field $\sigma$ leads to the relation
\begin{equation}
\label{N}
N=\frac1\beta\NPhi+\left(\frac{2s}\beta-1\right)(\g-1),
\end{equation}
and as in \eq{NNp} we assume $\NPhi$ and $2s$ are both multiples of $\beta$.
Next, we require that the action and the vertex operators in \eq{pathi} are single-valued
under the identification $\sigma\sim\sigma+2\pi R_c$,
\begin{equation}\nonumber
\sqrt\beta R_c\in\mathbb Z,\quad \frac1{\sqrt\beta}\NPhi R_c +q(1-\g)R_c\in\mathbb Z.
\end{equation}
Hence the following choice of the compactification radius 
\be\label{compr}
R_c=\sqrt\beta
\ee
can be assumed.

The goal of this section is to derive an explicit formula for $\mathcal V\bigl(g,B,\{z_n\}\bigr)$ for any choice of the metric and magnetic field. Then in Prop.\ \ref{nugB} we show that for a special choice of the metric (Arakelov metric \eq{garak}) and magnetic field \eq{magnchoice} the path integral \eq{pathi} is a generating functional for the Laughlin states on this background geometry.

\subsection{Instanton configurations}

We begin by splitting the free field $\sigma$ into the classical part $\sigma_{ cl}$ and quantum part $\sigma_{qu}$,
$$
\sigma(z,\bz)=\sigma_{cl}(z,\bz)+\sigma_{qu}(z,\bz).
$$ 
Here $\sigma_{qu}(z,\bz)$ is single-valued on $\Sigma$, while the classical part $\sigma_{cl}(z,\bz)$ is multi-valued. The latter encodes all nontrivial field configurations on $\Sigma$ satisfying $\sigma_{cl}\sim\sigma_{cl}+2\pi R_c\mathbb Z$, when taken around any topologically nontrivial one-cycle on the Riemann surface. 

Classical field configurations are labelled by two sets of integers $m,m'\in\mathbb Z^\g$, $\sigma_{cl}=\sigma_{mm'}$. The subscript $mm'$ indicates the configuration winding $m_j$, resp.\ $m_j'$ times under continuation along the $a_j$-, resp.\ $b_j$-cycle  
\begin{equation}\label{cuts}
\sigma_{mm'}(z+a_j)=\sigma_{mm'}(z)+2\pi R_cm_j,\quad {\rm resp.} \quad 
\sigma_{mm'}(z+b_j)=\sigma_{mm'}(z)-2\pi R_cm'_j.
\end{equation} 
This classifies all nontrivial maps of $\Sigma$ into a circle, also called instanton configurations. The general solution to \eqref{cuts} for a given marking of $\Sigma$ can be written down in terms of the harmonic one-forms \eqref{harmf} as follows 
\begin{equation}\nonumber
\sigma_{mm'}(z,\bz)=2\pi R_c\sum_{j=1}^\g\left(m_j\int_{P_0}^z\alpha_j-m'_j\int_{P_0}^z\beta_j\right),
\end{equation}
where $P_0\in\Sigma$ is the base point.
Using \eq{alpha1} we can also express $\sigma_{mm'}$ in terms of the holomorphic differentials and the period matrix 
\begin{equation}\label{inst}
\sigma_{mm'}(z,\bz)=\pi i R_c\left((m'+m\bar\tau)_j{(\Im\tau)^{-1}}_{jl}\int_{P_0}^z\omega_l-(m'+m\tau)_j{(\Im\tau)^{-1}}_{jl}\int_{P_0}^z\bar\omega_l\right),
\end{equation}
where as usual the sum over repeated indices $j,l$ is understood.

The quantum part of the field $\sigma_{qu}$ splits 
$$
\sigma_{qu}(z,\bz)=\sigma_0+\tilde\sigma(z,\bz)
$$
into the constant mode $\sigma_0$ and fluctuating component $\tilde\sigma(z,\bz)$ orthogonal to $\sigma_0$ in the $L^2$ metric,
$$
\int_\Sigma\tilde\sigma\sqrt gd^2z=0.
$$
The measure splits as $D_g\sigma=d\sigma_0D_g\tilde\sigma$ and the integration over $\sigma_0$ yields the constraint \eqref{N}. Since the action is the sum of the classical and quantum parts, $S(\sigma)=S(\sigma_{mm'})+S(\sigma_{qu})$, the path integral \eq{pathi} splits as a product 
\be\label{productV}
\mathcal V\bigl(g,B,\{z_n\}\bigr)=Z_{\rm cl}\cdot Z_{\rm qu},
\ee
Here $Z_{\rm cl}$ is the sum over all instanton configurations
\begin{equation}\label{instsum}
Z_{\rm cl}=\left(\frac{R_c}{\sqrt 2}\right)^\g\sum_{m,m'\in\mathbb Z^\g}e^{i\sqrt{\beta}\sum_{n=1}^N\sigma_{mm'}(z_n,\bz_n)-S(\sigma_{mm'})},
\end{equation}
where we introduced a convenient normalization constant, and $Z_{\rm qu}$ will be computed in \S \ref{Zq}. Before evaluating the sum above we need to compute $S(\sigma_{mm'})$.

\subsection{Linear terms in the action}
Since $\sigma_{mm'}(z,\bz)$ is multi-valued on $\Sigma$ the linear terms in the action \eq{action} are not well-defined as they stand. The proper definition was given in Refs.\ \cite{VV,ABMNV}. Here we adopt with minor changes an equivalent definition proposed in Ref.\ \cite{GMMOS}, since it will be suitable for the magnetic field term in \eq{action} as well.

Consider the canonical dissection of the Riemann surface, where one contracts the $a,b$-cycles to the base point $P_0$ and cuts the surface along the resulting cycles, as illustrated in Fig.\ \ref{fig:dissect} below for a surface of genus two.

\begin{figure}[h]
\begin{center}
\begin{tikzpicture}[smooth cycle,scale=0.45]
\tikzset{->-/.style={decoration={
  markings,
  mark=at position .5 with {\arrow{>}}},postaction={decorate}}}
  \tikzset{-<-/.style={decoration={
  markings,
  mark=at position .5 with {\arrow{<}}},postaction={decorate}}}
\draw [black,very thick,bend right] (6,6.35) arc (50:310:3.5);
\draw [black,very thick,bend right] (10,1) arc (-130:130:3.5);
\draw [black,very thick,bend right] (10.01,0.99) edge (5.96,.96);
\draw [black,very thick,bend right] (6,6.35) edge (10.04,6.39);
\draw [black,very thick,bend right] (2,3.7) edge (4.8,3.7);
\draw [black,very thick,bend left] (2.3,3.55) edge (4.5,3.55);
\draw [black,very thick,bend right] (11,3.7) edge (13.8,3.7);
\draw [black,very thick,bend left] (11.3,3.55) edge (13.5,3.55);
\draw [red,very thick,bend left,decoration={markings, mark=at position .65 with {\arrow{<}}},postaction={decorate}] (4.45,3.5) ellipse (100pt and 37pt);
\draw [blue,very thick,bend right] (8.05,3.6) edge [->-] (2.5,6.94);
\draw [blue,very thick,bend left]  (3,3.82) edge [->-] (8,3.6);
\draw [green,very thick,bend left,decoration={markings, mark=at position .15 with {\arrow{<}}},postaction={decorate}] (4.55+7,3.5) ellipse (100pt and 37pt);
\draw [cyan,very thick,bend left] (8,3.6) edge [-<-] (11.4,3.55);
\draw [cyan,very thick,bend left] (9.5,1.25) edge [-<-] (8,3.6);
\draw [blue,dashed,very thick,bend right] (2.5,6.94) edge (3,3.82);
\draw [cyan,dashed,very thick,bend left] (11.4,3.55) edge (9.5,1.25);
        \draw[black] (8.15,3.95) node [above]{$P_0$};
        \fill [black] (8,3.6) circle (3pt);
  \end{tikzpicture}\quad\quad
\begin{tikzpicture}[scale=0.6]
\tikzset{->-/.style={decoration={
  markings,
  mark=at position .5 with {\arrow{>}}},postaction={decorate}}}
  \tikzset{-<-/.style={decoration={
  markings,
  mark=at position .5 with {\arrow{<}}},postaction={decorate}}}
\draw [red,very thick,bend left] (0,0) edge [->-] (2,1);
\draw [blue,very thick,bend left] (2,1) edge [->-] (3,3);
\draw [red,very thick,bend left] (3,3) edge [-<-] (2,5);
\draw [blue,very thick,bend left] (2,5) edge [-<-] (0,6);
\draw [green,very thick,bend left] (0,6) edge [->-] (-2,5);
\draw [cyan,very thick,bend left] (-2,5) edge [->-] (-3,3);
\draw [green,very thick,bend left] (-3,3) edge [-<-] (-2,1);
\draw [cyan,very thick,bend left] (-2,1) edge [-<-] (0,0);
    \draw[black] (0,3) node[]{$\Sigma_0$};
    \draw[black] (1.3,0.9) node[below]{$b_1^+$};
        \draw[black] (2.7,2.7) node[below]{$a_1^+$};
    \draw[black] (2.5,4) node[above]{$b_1^-$};
    \draw[black] (.8,5.3) node[above]{$a_1^-$};
        \draw[black] (-1.5,5.2) node[above]{$b_2^+$};
    \draw[black] (-2.8,3.4) node[above]{$a_2^+$};
    \draw[black] (-2.4,2) node[below]{$b_2^-$};
        \draw[black] (-.7,.7) node[below]{$a_2^-$};
  \end{tikzpicture}
{\small \caption{Canonical dissection of the genus-2 Riemann surface.}
\label{fig:dissect}}
\end{center}
\end{figure}
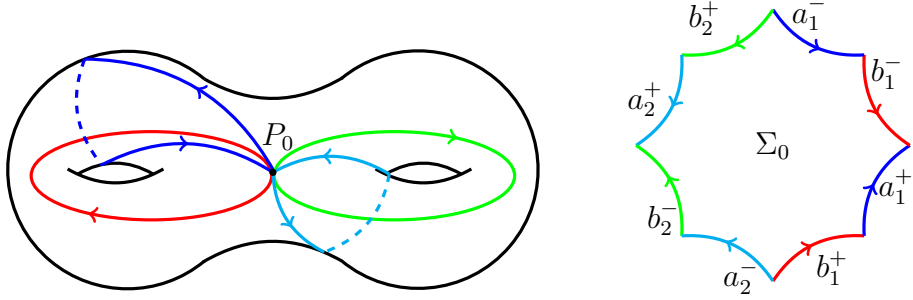

The resulting simply-connected 2-cell $\Sigma_0$ has the boundary
\begin{equation}\label{dsigma}
\p \Sigma_0=\sum_{j=1}^\g\bigl(a_j^++b_j^+-a_j^--b_j^-\bigr),
\end{equation}
where $a_j^+,b_j^+$ (resp. $a_j^-,b_j^-$) are left (resp. right) sides of cuts along cycles $a_j,b_j$.

Then the curvature term in the action is defined as integral over $\Sigma_0$ with the correction term on the boundary
\begin{equation}
\label{Rsigma}
\frac1{8\pi}\int_\Sigma\sigma_{mm'} R\sqrt gd^2z:=\frac1{8\pi}\int_{\Sigma_0}\sigma_{mm'} R\sqrt gd^2z
+\frac1{4\pi}\int_{\p\Sigma_0}\sigma_{mm'}\Omega,
\end{equation}
where we introduced the one-form $\Omega$, satisfying
\begin{equation}\label{Omega}
d\Omega=-R_{z\bz}\,idz\wedge d\bz-4\pi\sum_{\alpha=1}^{\g-1}\delta(z-p_\alpha)\,\frac i2dz\wedge d\bz.
\end{equation}
The delta functions are supported on a positive divisor $\tilde D$ of degree $\g-1$,
\begin{equation}
\label{tildeD}
\tilde D=\sum_{\alpha=1}^{\g-1} p_\alpha,
\end{equation} 
such that its doubling $2\tilde D\equiv C$ is in the canonical class. This can be taken as one of the divisors $D_\delta$ corresponding to an odd theta-caracteristic, as described in \eq{Ddelta}.
It is straightforward to see that $\int_{\Sigma_0}d\Omega=0$.

Now, let $\eta(z)$ be a holomorphic section of the line bundle $\mathcal O(D_\delta)$ with the divisor of zeroes $(\eta)=D_\delta$. Then we can write \eq{Omega} in the form
\begin{equation}\label{domega}
d\Omega=dd^c\log\frac{g_{z\bz}^{1/2}}{|\eta(z)|^2}
\end{equation}
where $dd^c\cdot=2idz\wedge d\bz\;\p_z\p_{\bz}\;\cdot$, following \eq{ddc}, and we used the relation 
\begin{equation}\label{deltafunc}
\p_z\p_{\bz}\log|z-\alpha|^2=\pi\delta(z-\alpha)
\end{equation} 
in a local complex coordinate $z$ around each $p_\alpha$. Then we can solve for $\Omega$ as
\begin{equation}\nonumber
\Omega=d^c\log\frac{g_{z\bz}^{1/2}}{|\eta(z)|^2},
\end{equation}
with an ambiguity of adding an arbitrary closed differential, which we here set to zero, see the discussion later in \S \ref{magnft}.

\begin{prop}\label{indep1}
The linear term as defined in \eq{Rsigma} depends only on the divisor class $[\tilde D]$ and is independent of the choice of the metric $g$ in the same conformal class. 
\end{prop}
\begin{proof}
Changing the divisor $\tilde D=(\eta)$ to $\tilde D'=(\eta')$ in the same class $\tilde D\equiv\tilde D'$ the ratio $\eta'/\eta$ is a meromorphic function on $\Sigma$.
The linear term \eq{Rsigma} then changes as
\begin{multline}\nonumber
\frac{1}{8\pi}iq\int_\Sigma\sigma_{mm'} R\sqrt gd^2z\big|_{\tilde D}-\frac{1}{8\pi}iq\int_\Sigma\sigma_{mm'} R\sqrt gd^2z\big|_{\tilde D'}=\frac1{4\pi}iq\int_{\p\Sigma_0}\sigma_{mm'}d^c\log\left|\frac{\eta'}{\eta}\right|^2\\
=\frac1{4\pi}iq\sum_{j=1}^\g\left(\int_{a_j}\Delta_{b_j}\sigma_{mm'}\,d^c\log\left|\frac{\eta'}{\eta}\right|^2-\int_{b_j}\Delta_{a_j}\sigma_{mm'}\,d^c\log\left|\frac{\eta'}{\eta}\right|^2\right)\\
=\frac{iqR_c}{2}\sum_{j=1}^\g\left(m_j'\int_{a_j}d^c\log\left|\frac{\eta'}{\eta}\right|^2+m_j\int_{b_j}d^c\log\left|\frac{\eta'}{\eta}\right|^2\right)\\
=\frac{iqR_c}{2}\sum_{j=1}^\g\left(m_j'\int_{a_j}\left(i\p_z\log \frac{\eta'}{\eta}\,dz+c.c\right)+m_j\int_{b_j}\left(i\p_z\log \frac{\eta'}{\eta}\,dz+c.c\right)\right)=2\pi iR_cq\cdot\mathbb Z,
\end{multline}
and $qR_c\in\mathbb Z$, \eq{compr}.
In the second line the jumps $\Delta_{a_j}\sigma_{mm'}$ and $\Delta_{b_j}\sigma_{mm'}$ are given by \eq{cuts}. In the last line we use the fact that the meromorphic function $\eta'/\eta$ can only change by a factor $e^{2\pi in},n\in\mathbb Z$ under continuation along any closed one-cycle.

Next we show independence on the choice of the metric $g$ in the same conformal class. Let $g'_{z\bz}=e^{\lambda(z,\bz)}g_{z\bz}$ is the new metric in the same class, i.e. $\lambda(z,\bz)$ is a global real-valued scalar function on $\Sigma$. We have
\begin{multline}\nonumber
\frac1{8\pi}\int_\Sigma\sigma_{mm'} R\sqrt gd^2z-\frac1{8\pi}\int_\Sigma\sigma_{mm'} R(g')\sqrt{g'}d^2z\\=\frac1{8\pi}\int_{\Sigma_0}\sigma_{mm'} dd^c\lambda
-\frac1{8\pi}\int_{\p\Sigma_0}\sigma_{mm'} d^c\lambda
=\frac1{8\pi}\int_{\Sigma_0}d\left(\sigma_{mm'} d^c\lambda\right)\\
-\frac1{8\pi}\int_{\Sigma_0}d\sigma_{mm'} \wedge d^c\lambda-\frac1{8\pi}\int_{\p\Sigma_0}\sigma_{mm'} d^c\lambda
=-\frac1{8\pi}\int_{\Sigma_0}d\sigma_{mm'}\wedge d^c\lambda\\=
-\frac1{8\pi}\sum_{j=1}^\g\int_{\Sigma_0}(c_j\omega_j+\bar c_j\bar\omega_j)\wedge d^c\lambda
=\frac1{8\pi}\sum_{j=1}^\g\int_{\Sigma_0}d^c(\lambda c_j\omega_j+\lambda\bar c_j\bar\omega_j)=\\
\frac1{8\pi}\sum_{j=1}^\g\int_{\Sigma_0}d(i\lambda c_j\omega_j-i\lambda\bar c_j\bar\omega_j)=
\frac1{8\pi}\sum_{j=1}^\g\int_{\p\Sigma_0}(i\lambda c_j\omega_j-i\lambda\bar c_j\bar\omega_j)=0.
\end{multline}
Here in the third line we make use of the relation $d\sigma_{mm'}=\sum_{j=1}^\g c_j\omega_j+\bar c_j\bar\omega_j$, where the complex numbers $c_j,\bar c_j$ can be inferred from \eq{inst}. In the last line we use the fact that $\lambda$ is a global scalar function, i.e., it has no jumps along the cycles and the integrals over $a_j^{+},b_j^{+}$ segments of the boundary cancel out with those over $a_j^{-},b_j^{-}$, due to the direction reversal in \eq{dsigma}.
\end{proof}

Now we can use this proposition to compute \eq{Rsigma}. 
First we choose the divisor corresponding to a particular odd half characteristic $\delta$. Without loss of generality we can use the one introduced in \eq{degDd},
\begin{equation}\label{canon}
\tilde D=D_\delta.
\end{equation}
By \eq{del}, $\delta=I[D_\delta]+\Delta$, where $\Delta$ is the vector of Riemann constants and there exists a holomorphic differential $\omega_\delta$ with the divisor of zeroes $(\omega_\delta)=2D_\delta$. Then we can construct the following singular Riemannian metric $g_{z\bz}^\delta$ on $\Sigma$
\begin{equation}\nonumber
2g_{z\bz}^\delta dzd\bz=2|\omega_\delta|^2,
\end{equation}
called the Mandelstam metric in Ref.\ \cite{DP2}.
With these choices $\Omega=0$ vanishes and we immediately obtain
\begin{equation}
\label{Rsigma1}
\frac1{8\pi}q\int_\Sigma\sigma_{mm'} R\sqrt gd^2z=
-q\sigma_{mm'}(D_\delta)+2\pi \mathbb Z.
\end{equation}
where again $D_\delta$ is the divisor \eqref{degDd}.

\subsection{Magnetic field term}
\label{magnft}

Now we treat the magnetic field term, i.e., the third term in the action \eq{action} along the same lines. We introduce the gauge connection one-form for the magnetic field $(1,1)$ form \eq{curvF},
\be\nonumber
F:=dA=(\p_zA_{\bz}-\p_{\bz}A_z)dz\wedge d\bz
\ee
and define
\begin{equation}\label{defB}
\frac{1}{2\pi}\int_\Sigma \sigma_{mm'} B\sqrt gd^2z:=\frac1{2\pi}\int_{\Sigma_0}\sigma_{mm'} B\sqrt gd^2z-\frac1{2\pi}\int_{\p\Sigma_0}\sigma_{mm'}\Lambda.
\end{equation}
Here the one-form $\Lambda$ satisfies
\begin{equation}\label{dLambda}
d\Lambda=F-2\pi\sum_{a=1}^\NPhi\delta(z-q_a)\frac i2dz\wedge d\bz,
\end{equation}
where points $q_a$ belong to the magnetic field divisor $D_\m$, \eq{degDm}.

In contrast to the curvature term \eqref{Rsigma}, where we fixed the line bundle $D_\delta$, here we would like to parameterize all line bundles in ${\rm Pic}_\NPhi(\Sigma)$. 
Let $s(z)$ be a holomorphic section of $\mathcal O(D_\m)$ with the divisor of zeroes $D_\m$. Then we can solve for $\Lambda$ in \eq{dLambda} as follows
\begin{equation}\label{alpha}
\Lambda=-\frac12d^c\log \bigl(h^{\NPhi}(z,\bz)|s(z)|^2\bigr)+\gamma,
\end{equation}
where $\gamma$ is an arbitrary closed differential and $h^\NPhi$ is the hermitian metric \eq{curvF}. We will now identify $\gamma$ with the flat connection part of the gauge field $A$, with nontrivial monodromies around the one-cycles. Then the moduli of flat connections will provide coordinates in ${\rm Pic}_\NPhi$.  

Note that the definition \eq{defB} is obviously insensitive to the shift of $\Lambda$ by an exact one-form $\Lambda\to\Lambda+df$. Hence we can parameterize an arbitrary closed differential $\gamma\in H^1(\Sigma,\mathbb R)$ \eqref{alpha} as a linear combination of harmonic one-forms
\begin{equation}\label{flatconn}
\gamma=2\pi\sum_{j=1}^{\rm g}(\varphi_{1j}\alpha_j-\varphi_{2j}\beta_j),
\end{equation}
with real coefficients $\varphi_{1j},\varphi_{2j}$.
This leads to the monodromies $e^{2\pi i\varphi_{1j}}$ around the cycle $a_j$ and $e^{-2\pi i\varphi_{2j}}$ around $b_j$. Since we phases are defined mod $2\pi i\, \mathbb Z$, we obtain that
$$(\varphi_{1j},\varphi_{2j})\in(\mathbb R/\mathbb Z)^{2\rm g}$$ 
take values in a $2\rm g$-dimensional torus. 

Using \eq{alpha1} we now recast \eq{flatconn} in terms of holomorphic differentials
\begin{equation}\nonumber
\gamma=\pi i\bar\varphi(\Im\tau)^{-1}\omega-\pi i\varphi(\Im\tau)^{-1}\bar\omega,
\end{equation}
and we identify the phases with the complex coordinate on $Jac(\Sigma)$
\begin{equation}\label{complphi}
\varphi=\varphi_{2}+\tau\varphi_{1}.
\end{equation}

The definition of the magnetic field term \eq{defB} is independent of the choice of the hermitian metric $h^{\NPhi}$ and depends only on the divisor class of $D_\m$. The proof is completely analogous to the proof of Prop.\ \ref{indep1} and we do not repeat it here. Now we are ready to evaluate the linear term in the action \eq{defB}.

Independence of the choice of $h^{\NPhi}$ allows us to choose a convenient hermitian metric on $L^{\NPhi}$. Let $\tilde s(z)$ be an arbitrary  holomorphic section $L^{\NPhi}$ and consider the singular metric $h^\NPhi=1/|\tilde s|^2$. The curvature \eqref{curvF} of this metric is supported on the divisor ${\rm div}\,\tilde s=(\tilde s)$ of zeroes of $\tilde s(z)$. Since this divisor is in the equivalence class of $(\tilde s)\equiv D_\m$, we obtain,
\begin{multline}\label{defB1}
\frac{1}{2\pi}\int_\Sigma \sigma_{mm'} B\sqrt gd^2z=\sigma_{mm'}({\rm div}\,\tilde s)-\frac1{2\pi}\int_{\p\Sigma_0}\sigma_{mm'}\gamma\\
=\sigma_{mm'}({\rm div}\,\tilde s)-2\pi R_c\sum_{j=1}^\g\int_{\p\Sigma_0}\left(\gamma(z)\int_{P_0}^z(m_j\alpha_j-m'_j\beta_j)\right)\\
=\sigma_{mm'}(D_\m)-2\pi R_c\left(m'\varphi_1-m\varphi_2\right)+2\pi R_c \mathbb Z
\end{multline}
Here in the second line we use the Riemann bilinear relation \eq{r1} in the form \cite[III.3]{FKra}, 
\be\nonumber
\int_{\p \Sigma_0} \left(\gamma_2(P)\int_{P_0}^P\gamma_1\right)=\sum_{j=1}^\g\;\int_{a_j}\gamma_1\cdot\int_{b_j}\gamma_2-\int_{b_j}\gamma_1\cdot\int_{a_j}\gamma_2.
\ee
Then apply the Abel theorem,
$I[D'-D]=0\;\;{\rm mod}\,\Lambda,$ iff $D\equiv D'$, for $D=(\tilde s)$ and $D'=D_\m$
\begin{multline}\nonumber
\sigma_{mm'}({\rm div}\,\tilde s)-\sigma_{mm'}(D_{\rm m})\\=\pi i R_c\left((m'+m\bar\tau)_j{(\Im\tau)^{-1}}_{jl}\int_{D_\m}^{{\rm div}\,\tilde s}\omega_l-(m'+m\tau)_j{(\Im\tau)^{-1}}_{jl}\int_{D_\m}^{{\rm div}\,\tilde s}\bar\omega_l\right)\\=2\pi R_c\mathbb Z.
\end{multline}
Finally we can recast the last term in \eq{defB1} in term of the complex coordinate \eqref{complphi},
$$-2\pi R_c\left(m'\varphi_1-m\varphi_2\right)=
\sigma_{mm'}(D_\m)+\pi iR_c\left((m'+m\bar\tau)(\Im\tau)^{-1}\varphi-(m'+m\tau)(\Im\tau)^{-1}\bar\varphi\right).
$$
Now we are in a position to calculate the instantonic sum \eqref{instsum}.

\subsection{Computing the path integral}
\label{Zq}

First we compute the instantonic contribution to the quadratic term in the action
\begin{multline}\label{kinetic}
\frac1{2\pi}\int_\Sigma\p_z\sigma_{mm'}\p_{\bz}\sigma_{mm'} d^2z=\frac{i}{4\pi}\int_\Sigma\p_z\sigma_{mm'} dz\wedge\p_{\bz}\sigma_{mm'} d\bz\\
=\frac{\pi i}4R_c^2(m'+m\bar\tau)_j{(\Im\tau)^{-1}}_{jl}(m'+m\tau)_k{(\Im\tau)^{-1}}_{kp}
\int_\Sigma\omega_l\wedge\bar\omega_p\\
=\frac{\pi}2R_c^2(m'+m\bar\tau)_j{(\Im\tau)^{-1}}_{jl}(m'+m\tau)_l,
\end{multline}
and the summation over the indices $j,k,l,p=1,...,\g$ is understood.
In the second line we used 
$$\int_\Sigma\omega_l\wedge\bar\omega_p=-2i\,\Im\tau_{pl},$$
which follows from the Riemann's bilinear relation \eq{r1} and \eq{acycl}.

Putting together Eqns.\ (\ref{Rsigma1}, \ref{defB1}, \ref{kinetic}) we arrive at the following
expression for the sum \eq{instsum} over instanton configurations,
\begin{multline}\label{doublesum}
Z_{\rm cl}=\left(\frac{R_c}{\sqrt 2}\right)^\g\sum_{m,m'\in\mathbb Z^\g}\exp\left\{-\frac{\pi}2R_c^2(m'+m\bar\tau){(\Im\tau)^{-1}}(m'+m\tau)+i\sqrt{\beta}\sum_{n=1}^N\sigma_{mm'}(z_n,\bz_n)\right.
\\\left.+qi\sigma_{mm'}(D_\delta)-i\frac1{\sqrt{\beta}}\sigma_{mm'}(D_\m)
+\pi R_c\frac1{\sqrt{\beta}}\left((m'+m\bar\tau)(\Im\tau)^{-1}\varphi-(m'+m\tau)(\Im\tau)^{-1}\bar\varphi\right)\right\}\\
=\left(\frac{R_c}{\sqrt 2}\right)^\g\sum_{m,m'\in\mathbb Z^\g}\exp\left\{-\frac{\pi}2R_c^2(m'+m\bar\tau){(\Im\tau)^{-1}}(m'+m\tau)\right.
\\\left.-\pi\left((m'+m\bar\tau)(\Im\tau)^{-1}D-(m'+m\tau)(\Im\tau)^{-1}D\right)\right\}.
\end{multline}
Here $D$ is a shortcut notation for the following divisor
\begin{equation}\label{divD}
D=\beta \sum_{n=1}^Nz_n+\beta D_\delta-{\rm div}\,\L.
\end{equation}
By the charge conservation condition \eq{N} this divisor has degree zero and therefore $Z_{\rm cl}$ is independent of the choice of the base point $P_0$.

Next step in the calculation is applying the Poisson summation formula to the sum over $m'$ in \eq{doublesum}. This is a standard calculation which goes along the same lines as in, e.g., Ref.\ \cite{ABMNV,GMMOS}, and here we state the final result  
\begin{equation}\label{Zcl1}
Z_{\rm cl}=\frac1{2^\g}\sqrt{\det\Im\tau}\cdot e^{-\frac{2\pi}\beta\Im D\,(\Im \tau)^{-1}\Im D}\sum_{\varepsilon',\varepsilon''\in\{0,\frac12\}^\g}\;\sum_{r\in(\mathbb N^+_{\leqslant\beta})^\g}\;e^{4\pi i(\varepsilon',\varepsilon'')}\left|\vartheta\left[\begin{array}{c}\scriptstyle\frac r\beta+\varepsilon'\\\scriptstyle\varepsilon''\end{array}\right](D,\beta\tau)\right|^2
\end{equation}
Here the notation $r\in(\mathbb N^+_{\leqslant\beta})^\g$ means that $r$ takes values in $\g$ copies of the string of integers $(1,...,\beta)$. 

Instead of \eq{divD} it will be convenient to use the following notation 
\be\label{divD0}
D_0=\beta \sum_{n=1}^Nz_n-\beta \Delta-{\rm div}\,\L,
\ee
where $\Delta=\Delta^{P_0}$ is the vector of Riemann constants \eq{riemannC}.
Written in terms of $D_0$, \eq{Zcl1} has the form
\begin{align}\nonumber
Z_{\rm cl}=\frac1{2^\g}\sqrt{\det\Im\tau}\,\cdot\,& e^{-\frac{2\pi}\beta\Im D_0\,(\Im \tau)^{-1}\Im D_0}\\&\label{Zcl2}
\cdot\sum_{\varepsilon',\varepsilon''\in\{0,\frac12\}^\g}\;\sum_{r\in(\mathbb N^+_{\leqslant\beta})^\g}\;e^{4\pi i(\varepsilon',\varepsilon'')}\left|\vartheta\left[\begin{array}{c}\scriptstyle\frac r\beta+\varepsilon'+\delta' \\\scriptstyle\varepsilon''+\beta\delta''\end{array}\right](D_0,\beta\tau)\right|^2
\end{align}
with shifted theta characteristics. 

Now we turn to computing $Z_{\rm qu}$, which is the regularized infinite-dimensional gaussian integral of the free field $\tilde\sigma$ with zero mean. The computation is the same for the surfaces any genus $\g$ and here we closely follow Ref.\ \cite[\S 4.3]{K16}. The Green function for the scalar laplacian satisfies
\begin{align}
\label{green}
&\Delta_g G_{g}(z,y)=-2\pi\delta_g(z,y)+\frac{2\pi}{A_g},\\\label{green1}
&\int_\Sigma G_{g}(z,y)\sqrt{g}d^2y=0,
\end{align}
where $\delta_g(z,y)$ is defined so that $\int_\Sigma \delta_g(z,y)\sqrt gd^2z=1$.
The regularized Green function at coincident points is defined as
\begin{equation}
\label{reggreen}
G_{g}^{\rm reg}(z)=\lim_{z\to y}\bigl(G_{g}(z,y)+\log d_{g}(z,y)\bigr),
\end{equation}
where $d_{g}(z,y)$ is the geodesic distance between $y$ and $z$ in the metric $g$. For small distances $d_{g}(z,y)=\sqrt{g_{z\bz}}|z-y|+ \mathcal O(|z-y|^2)$ in local complex coordinates.

The quantum part of the path integral $Z_{\rm qu}$ has the form of gaussian integral with the (purely imaginary) linear term. Hence we can use the standard formula
\begin{equation}\label{zqu}
\int e^{-\frac1{2\pi}\int_{\Sigma}(\p_z\tilde\sigma\p_{\bz}\tilde\sigma+i\tilde\sigma j\sqrt{g})d^2z}\mathcal D_g\tilde\sigma=\left[\frac{\det'\Delta_g}{A_g}\right]^{-1/2}e^{-\frac1{4\pi^2}\int_{\Sigma\times\Sigma}j(z)G_g(z,z')|_{\rm reg}j(z)\sqrt g d^2z\sqrt gd^2z'},
\end{equation}
where the subscript reg indicates that for the integration along the diagonal in $\Sigma\time\Sigma$ we assume the coincident-point regularization \eq{reggreen}. In the case of path integral \eq{pathi} the current $j$ has the form
\begin{equation}\nonumber
j(z)=\frac{q}{4} R(z)+\frac 1{\sqrt\beta} B(z)-2\pi \sqrt\beta\sum_{n=1}^{N}\delta_g(z,z_n).\end{equation}

Plugging this into the general formula \eq{zqu} we arrive at the formula for $Z_{\rm qu}$,
\begin{multline}\label{zqu1}
Z_{\rm qu}=\left[\frac{\det'\Delta_g}{A_g}\right]^{-1/2}\cdot\exp\left(-\beta\sum_{n\neq m}^{N}G_{g}(z_n,z_m)-\beta\sum_{n=1}^{N}G_{g}^{\rm reg}(z_n)\right.\\\left.
+\frac{\sqrt\beta}{\pi}\sum_{n=1}^{N}\int_\Sigma G_{g}(z_n,z)\left(\frac q4R+\frac1{\sqrt\beta}B\right)\big|_z\sqrt{g}d^2z
\right.\\\left.-\frac1{4\pi^2}\int_{\Sigma\times\Sigma}\left(\frac q4R+\frac1{\sqrt\beta}B\right)\big|_zG_{g
}(z,z')\left(\frac q4R+\frac1{\sqrt\beta}B\right)\big|_{z'}\sqrt{g}d^2z\,\sqrt{g}d^2z'\right).
\end{multline}
Here the regularization prescription amounts to replacing divergent $G_g(z_n,z_n)$ on the diagonal by the regularized Green function $G_{g}^{\rm reg}(z_n)$.

The final result for the \eq{productV} is the product of $Z_{\rm cl}$ and $Z_{\rm qu}$. 
We have shown that the classical part $Z_{\rm cl}$ of the partition function is topological, i.e., independent of the choice of a Riemannian metric on $\Sigma$ and hermitian metric on $L_\NPhi$. Therefore all dependence on the metric and magnetic field is in the quantum part
$$Z_{\rm qu}:=Z_{\rm qu}[g,B].$$ 

In what follows we make a convenient choice of the Riemannian metric and magnetic field, where the path integral can be directly identified as a generating functional for the Laughlin states.

\subsection{Green functions and metrics on $\Sigma$}
\label{arak}

Given a Green function for some fixed metric $g_0$, the Green function for an arbitrary metric $g=e^{\lambda(z,\bz)} g_0$ in the same conformal class can be computed according to the formula
\begin{align}\label{trangreen}\nonumber
G_g(z,y)=G_{g_0}(z,y)-\frac1{A_g}\int_\Sigma\big(G_{g_0}(z,z')&+G_{g_0}(z',y)\big)\sqrt gd^2z'\\&
+\frac1{A_g^2}\int_{\Sigma\times\Sigma}G_{g_0}(z',z'')\sqrt gd^2z'\sqrt gd^2z'',
\end{align}
which is a consequence of Eqns.\ (\ref{green}, \ref{green1}). Therefore it is enough to know $G_{g_0}$ for some convenient metric.

Consider the canonical metric  $g_{\rm c}$, \eq{canmet}. The Green function can be constructed with the help of the Prime form \eqref{pform}. The following object is a globally defined real-valued $(-\frac12,-\frac12)$-differential both in $z$ and $y$,
\begin{equation}\label{Fdiff}
F(z,y)=e^{-2\pi\Im (z-y)\,(\Im\tau)^{-1}\Im (z-y)}\;|E(z,y)|^2,
\end{equation}
where $\Im (z-y)$ stands for $\int_z^y\Im\omega$. Single-valuedness of $F(z,y)$ on $\Sigma\times\Sigma$ follows from the transformation properties \eqref{transE}.
From \eq{neard} and  \eq{deltafunc} we obtain
\be\nonumber
\p_z\p_{\bz}\log F(z,y)=\pi\delta(z-y)-\g\, {g_c}_{z\bz}.
\ee
Using this formula and plugging $-\frac12\log F(z,y)$ for $G_{g_0}(z,y)$ into \eq{trangreen} after some algebra we arrive at the expression for the Green function for the canonical metric
\begin{multline}
\label{greencan}
G_{g_c}(z,y)=-\frac1{8\pi^2}\int_{\Sigma\times\Sigma}\log\left(\frac{F(z,y)F(z',z'')}{F(z,z')F(y,z'')}\right)\sqrt{g_c}d^2z'\sqrt{g_c}d^2z''
\\=-\frac1{4\pi^2}\int_{\Sigma\times\Sigma}\left(\log\left|\frac{E(z,y)E(z',z'')}{E(z,z')E(y,z'')}\right|
+2\pi \Im (z'-z)\,(\Im\tau)^{-1}\Im (z''-y)\right)\sqrt{g_c}d^2z'\sqrt{g_c}d^2z'',
\end{multline}
and one can check that the conditions \eqref{green}, \eqref{green1} are satisfied.

Another distinguished metric on $\Sigma$ is the Arakelov metric $g_a$, see Refs.\ \cite{Ar,Fal,DS,DP2}. It defined from the condition that its Ricci curvature is proportional to the canonical metric,
\be\label{Armetdef}
R_{z\bz}(g_a)=-\p_z\p_{\bz}\log {g_a}_{z\bz}=2(1-\g){g_c}_{z\bz}.
\ee
This defines the metric $g_a$ only up to an overall constant, which can be fixed from the following considerations.

First we define for any metric $g$ the so-called Arakelov-Green function $G^A_g$ according to
\begin{align}\label{greenarak}
&\Delta_gG^A_g(z,y)=-2\pi\delta_g(z,y)+\frac {R(z)}{4(1-\g)},\\\nonumber
&\int_\Sigma G^A_g(z,y) R(y)\sqrt gd^2y=0.
\end{align}
The Arakelov-Green function is related to the usual Green function as follows
\begin{multline}\nonumber
G^A_g(z,y)=G_g(z,y)-\frac1{8\pi(1-\g)}\int_\Sigma G_g(z,z')R\sqrt gd^2z'-\frac1{8\pi(1-\g)}\int_\Sigma G_g(y,z')R\sqrt gd^2z'\\+\frac1{(8\pi)^2(1-\g)^2}\int_{\Sigma\times\Sigma} G_g(z',z'')R\sqrt gd^2z'R\sqrt gd^2z''.
\end{multline} 
Both Green functions have the same near-diagonal behavior. Therefore the regularized Arakelov-Green function on the diagonal can be defined similar to \eq{reggreen},
\begin{equation}
\label{Arreg}
G_{g}^{A,\rm reg}(z)=\lim_{z\to y}\bigl(G^A_{g}(z,y)+\log d_{g}(z,y)\bigr),
\end{equation}
When expressed in terms of the Arakelov-Green function, the quantum part of the path integral $Z_{\rm qu}$, \eq{zqu1} assumes a particularly compact form,
\be\label{zqu2}
Z_{\rm qu}=\left[\frac{\det'\Delta_g}{A_g}\right]^{-1/2}\cdot e^{-\beta\sum_{n\neq m}^{N}G^A_{g}(z_n,z_m)-\beta\sum_{n=1}^{N}G_{g}^{A,\rm reg}(z_n)-\frac1{4\pi^2\beta}\int_{\Sigma\times\Sigma}G^A_{g
}(z,z')B\sqrt{g}d^2z\,B\sqrt{g}d^2z'},
\ee
which again holds for the arbitrary choices of $B$ and $g$.

It immediately follows from \eq{greenarak} that
\be\label{greenarar}
G^A_{g_a}(z,y)=G_{g_c}(z,y),
\ee 
in other words, the Arakelov-Green function for the Arakelov metric equals the standard Green function for the canonical metric \eqref{greencan}. 

Now we can choose the normalization of the Arakelov metric \eqref{Armetdef} so that the regularized Arakelov-Green function \eq{Arreg} identically vanishes $G^{A,\rm reg}_{g_a}(z)=0$. Taking into account \eq{greenarar} and \eq{greencan} we obtain
\begin{multline}\nonumber
{g_a}_{z\bz}=\exp\Bigl(-2\lim_{z\to y}\big(G^A_{g_a}(z,y)+\log|z-y|\big)\Bigr)\\
=\exp\Bigl(-\frac1{\pi}\int_\Sigma\log F(z,z')\,\sqrt{g_c}d^2z'+\frac1{4\pi^2}\int_{\Sigma\times\Sigma}\log F(z',z'')\,\sqrt{g_c}d^2z'\sqrt{g_c}d^2z''\Bigr).
\end{multline}
Plugging this back to \eq{greencan} we derive the formula 
\be\label{ArG}
G^A_{g_a}(z,y)=-\frac12\log\Big({g_a}_{z\bz}(z)^\frac12{g_a}_{z\bz}(y)^\frac12\,F(z,y)\Big).
\ee 
Recall that $F(z,y)$ is the $(-\frac12,-\frac12)$-differential both in $z$ and $y$ \eqref{Fdiff}. Hence the expression inside the logarithm transforms as a scalar on $\Sigma\times\Sigma$.

The following useful formula for the Arakelov metric holds, see, e.g., Refs.\ \cite[Eq.\ 5.1.28]{GMMOS} and \cite[Eq.\ 1.31]{Fay92},
\be\label{garak}
{g_a}_{z\bz}(z,\bz)=c^{\,2}|\sigma(z)|^{\frac4\g}\;e^{\,\frac{4\pi}{\g(\g-1)}\Im\Delta^z\,(\Im\tau)^{-1}\Im\Delta^z},
\ee
where $\Delta^z$ is the vector of Riemann constants \eq{riemannCz} and $\sigma(z)$ is the holomorphic $\frac\g2$-differential introduced in \eq{sigmafunc}, not to be confused with the free field $\sigma(z,\bz)$. The constant $c$ here is $z$-independent, but depends on the surface $\Sigma$. For the explicit formula, which we will not need in this paper, we refer to Ref. \cite[Eq.\ 1.21]{Fay92}.

\subsection{Laughlin states from path integral}

Now we choose the metric on $\Sigma$ in the path integral to be the Arakelov metric $g_a$. The convenient choice of the magnetic field $(1,1)$ form is the canonical metric, as in \eq{Fzz}. For the magnetic field $B$ this translates into the relation
\be\label{magnchoice}
B_a=\frac{\NPhi}{8\pi(1-\g)}R(g_a),
\ee
where the constant is fixed by the quantization condition \eqref{quant} taking into account the Gauss-Bonnet formula $\int_\Sigma R\sqrt gd^2z=8\pi(1-\g)$. 

With this choice of the magnetic field the last term in the exponent in \eq{zqu2} vanishes and the path integral \eq{productV} has the following form
\begin{multline}\label{bos2}
\mathcal V\bigl(g_a,B_a,\{z_n\}\bigr)
=\left[\frac{\det'\Delta_{g_a}}{A_{g_a}\det\Im\tau}\right]^{-1/2}\cdot e^{-\beta\sum_{n\neq m}^{N}G^A_{g_a}(z_n,z_m)}\\\cdot e^{-\frac{2\pi}\beta\Im D_0\,(\Im \tau)^{-1}\Im D_0}\cdot\sum_{\varepsilon',\varepsilon''\in\{0,\frac12\}^\g}\;\sum_{r\in(\mathbb N^+_{\leqslant\beta})^\g}\;e^{4\pi i(\varepsilon',\varepsilon'')}\left|\vartheta\left[\begin{array}{c}\scriptstyle\frac r\beta+\varepsilon'+\delta'\\\scriptstyle\varepsilon''+\beta\delta''\end{array}\right](D_0,\beta\tau)\right|^2,
\end{multline}
where we use the notation \eq{Zcl2} and \eqref{divD0}.

Using \eq{ArG} for the Arakelov-Green function and the following relation
\begin{align}\nonumber
\sum_{n<m}^N\Im(z_n-z_m)\,(\Im\tau)^{-1}\Im(z_n-z_m)&
\\\nonumber
=-\Im \sum_{n=1}^Nz_n\,(\Im\tau)^{-1}\Im \sum_{n=1}^Nz_n&+N\sum_{n=1}^N\Im z_n\,(\Im\tau)^{-1}\Im z_n
\end{align}
we can extract a more detailed formula for the exponent in the first line in \eq{bos2},
\begin{multline}\nonumber
e^{-\beta\sum_{n\neq m}^{N}G^A_{g_a}(z_n,z_m)}
=e^{2\pi\beta\,\Im \sum_{n=1}^Nz_n\,(\Im\tau)^{-1}\Im \sum_{n=1}^Nz_n}\prod_{n<m}^N\big|E(z_n,z_m)\big|^{2\beta}\\\cdot\prod_{n=1}^Ne^{-2\pi\beta N\,\Im z_n\,(\Im\tau)^{-1}\Im z_n}\big({g_a}_{z\bz}(z_n,\bz_n)\big)^{\frac{\beta(N-1)}2}
\end{multline}
Now, using the explicit formula \eq{garak} for the $g_a$ and after simple, but tedious algebra we obtain 
\begin{multline}\nonumber
\mathcal V\bigl(g_a,B_a,\{z_n\}\bigr)\\
=\left[\frac{\det'\Delta_{g_a}}{A_{g_a}\det\Im\tau}\right]^{-1/2}\cdot c^{\,N(\beta N+2s-\beta)}\cdot e^{-\frac{2\pi}{\beta(\g-1)^2}\Im \left((\g-1)\,{\rm div}\, \L+\deg\L\cdot\Delta\right)\,(\Im\tau)^{-1}\Im \left((\g-1)\,{\rm div}\, \L+\deg\L\cdot\Delta\right)}\\
\cdot\frac1{2^\g}\sum_{\varepsilon',\varepsilon''\in\{0,\frac12\}^\g}\;\sum_{r\in(\mathbb N^+_{\leqslant\beta})^\g}\;e^{4\pi i(\varepsilon',\varepsilon'')}\left|\vartheta\left[\begin{array}{c}\scriptstyle\frac r\beta+\varepsilon'+\delta' \\\scriptstyle\varepsilon''+\beta\delta''\end{array}\right]\Big(\beta{\textstyle \sum_{n=1}^N} z_n-\beta\Delta-2sD_\delta-D_\m-\varphi,\beta\tau\Big)\right|^2\\
\cdot \prod_{n<m}^N\big|E(z_n,z_m)\big|^{2\beta}
\cdot\prod_{n=1}^N |\sigma(z_n)|^{\frac2\g(\beta N+2s-\beta)}\cdot\prod_{n=1}^N h^\NPhi(z_n,\bz_n)\big({g_a}_{z\bz}(z_n,\bz_n)\big)^{-s},
\end{multline}
where we use the hermitian metric $h^\NPhi$ defined in \eq{here}. 

Here we recognize the sum over the hermitian norms of the Laughlin states \eq{FL}. The additional structure is the graded sum over half-integer characteristics $(\varepsilon',\varepsilon'')$, relative to the reference odd spin structure $(\delta',\beta\delta'')$ corresponding to the divisor \eq{canon}. This can also be interpreted as twisting the flat line bundle in \eq{rmL} by half-integer points in $Jac(\Sigma)$. Thus we adopt the notation $F_r^{\varepsilon',\varepsilon''}$ for the corresponding Laughlin states,
\begin{multline}\label{laughlinst}
F_r^{\varepsilon',\varepsilon''}(z_1,...,z_N)=\vartheta\left[\begin{array}{c}\scriptstyle\frac r\beta+\varepsilon' \\\scriptstyle\varepsilon''\end{array}\right]\Big(\beta{\textstyle \sum_{n=1}^N} z_n-\beta\Delta-2sD_\delta-D_\m-\varphi,\beta\tau\Big)\\\cdot \prod_{n<m}^NE(z_n,z_m)^{\beta}
\cdot\prod_{n=1}^N \sigma(z_n)^{\frac1\g(\beta N+2s-\beta)}.
\end{multline}
This concludes the derivation of our final result, which we now summarize in the following proposition.
\begin{prop}\label{nugB}
The path integral \eq{pathi} computed at the Arakelov metric $g_a$ and for the choice of the magnetic field \eq{magnchoice} has the following form
\begin{multline}\label{final}
\mathcal V\bigl(g_a,B_a,\{z_n\}\bigr)\\
=\left[\frac{\det'\Delta_{g_a}}{A_{g_a}\det\Im\tau}\right]^{-1/2}\cdot c^{\,N(\beta N+2s-\beta)}\cdot e^{-\frac{2\pi}{\beta(\g-1)^2}\Im \left((\g-1)\,{\rm div}\, \L+\deg\L\cdot\Delta\right)\,(\Im\tau)^{-1}\Im \left((\g-1)\,{\rm div}\, \L+\deg\L\cdot\Delta\right)}\\
\cdot\frac1{2^\g}\sum_{\varepsilon',\varepsilon''\in\{0,\frac12\}^\g}\;\sum_{r\in(\mathbb N^+_{\leqslant\beta})^\g}\;e^{4\pi i(\varepsilon',\varepsilon'')}\Vert F_r^{\varepsilon'+\delta',\varepsilon''+\beta\delta''}(z_1,...,z_N)\Vert_h^2,
\end{multline}
where the Laughlin states are given by \eq{laughlinst} and notation for the point-wise hermitian norm for each $z_n$ is according to \eq{norm1}. 
\end{prop}

\begin{rem}
When $\beta=1$, the path integral  \eq{pathi} has a dual fermionic representation \cite{ABMNV,VV}, which we briefly recall here. One considers two independent complex-valued fermionic fields $b,c$, where $b$ transforms as (not necessarily holomorphic) section of the holomorphic line bundle $\L$ and $c$ transforms as the section of the dual bundle $\L^{-1}\otimes K$. The quadratic action is
$$
S=\frac1{2\pi}\int_\Sigma c\,\bp_{\L}b+\bar c\,{\bp_{\L}}^*\bar b.
$$
For $\deg \L$ large enough, $\L^{-1}\otimes K$ has no holomorphic sections and thus only the $b$-field has zero modes. Hence the non-zero correlation function will involve the insertions $N$ of fermion doublets $\Vert b\Vert^2_h=(b\bar b)\cdot h$ at positions $z_1,...,z_N$. Computing the fermionic path integral one obtains the result
\be\label{spinsferm}
\langle \Vert b(z_1)\Vert^2_h...\Vert b(z_N)\Vert^2_h\rangle=\frac{\det'\bp^+_\L\bp_\L^{\phantom{a}}}{\det\langle s_n,s_m\rangle_{L^2}}\,\Vert\det s_n(z_m)\Vert_h^2\,\text{.}
\ee
The equivalence with the bosonic correlation function \eq{pathi} is established according to the bosonization rule, saying that the fermion doublet is equivalent to the vertex operator in bosonic theory 
$$\Vert b(z)\Vert^2_h\sim e^{i\sigma(z,\bz)}.$$
One should also restrict to a fixed the spin structure sector $\delta$ in the bosonic path integral. Then
\be\nonumber
\left\langle \Vert b(z_1)\Vert^2_h\cdots\Vert b(z_N)\Vert^2_h\right\rangle_{\rm ferm}=\left\langle e^{i\sigma(z_1,\bz_1)}\cdots e^{i\sigma(z_N,\bz_N)}\right\rangle_{\rm bos}=\mathcal V_{\beta=1}\bigl(g_a,B_a,\{z_n\}\bigr)\big|_{\varepsilon=\delta}\,\text{.}
\ee
Taking $\beta=1$ and restricting to $\varepsilon=\delta$ in the path integral $\mathcal V$, \eq{bos2}, we reproduce bosonization formula \eq{bos1}. More explicitly, 
\begin{multline}\nonumber
\frac{\det'\bp^+_\L\bp_\L^{\phantom{a}}}{\det\langle s_n,s_m\rangle_{L^2}}\,\det s_n(z_m)\\=\left[\frac{\det'\Delta_{g_a}}{A_{g_a}\det\Im\tau}\right]^{-1/2}\cdot c^{\,N( N+2s-1)}\cdot e^{-\frac{2\pi}{(\g-1)^2}\Im \left((\g-1)\,{\rm div}\, \L+\deg\L\cdot\Delta\right)\,(\Im\tau)^{-1}\Im \left((\g-1)\,{\rm div}\, \L+\deg\L\cdot\Delta\right)}\\
\cdot\left|\vartheta\Big({\textstyle \sum_{n=1}^N} z_n-\Delta-2sD_\delta-D_\m-\varphi,\tau\Big)\right|^2
\cdot \prod_{n<m}^N\big|E(z_n,z_m)\big|^{2}
\cdot\prod_{n=1}^N |\sigma(z_n)|^{\frac2\g(N+2s-1)},
\end{multline}
where $N$ here is given by \eq{NNPhi}.
In this form the bosonization formula should be compared to \cite[Cor.\ 5.12]{Fay92}, taking into account a different normalization of the hermitian metric on $\L$ \cite[Eq.\ 1.36]{Fay92}
\end{rem}


\end{document}